\newif\ifdraft\drafttrue
\documentclass{article}
\usepackage[text={6.5in,9in},centering]{geometry}
\usepackage{amsmath,amssymb,latexsym,stmaryrd,scalerel,enumerate,tikz}
\usetikzlibrary{automata,backgrounds,shapes,arrows,positioning,calc,fit}
\usetikzlibrary{decorations,decorations.pathreplacing,decorations.pathmorphing,arrows.meta}
\usepackage[colorlinks,citecolor=red]{hyperref}
\usepackage{apxproof,cite}
\usepackage{authblk}

\numberwithin{equation}{section}

\newtheoremrep{theorem}{Theorem}[section]
\newtheoremrep{lemma}[theorem]{Lemma}

\newcommand\sem[1]{\llbracket#1\rrbracket}
\newcommand\amp{\mathrel\&}

\DeclareMathOperator*{\bigamp}{\scalerel*{\amp}{\textstyle\sum}} 
\newcommand\RAlg{R\hyphen\mathsf{Alg}}

\newcommand\eval{\mathsf{eval}}
\newcommand\beps{{\pmb{\eps}}}

\newcommand\NS{\naturals^{\Sigma^*}}
\newcommand\Ns{\naturals^S}
\newcommand\Nt{\naturals^T}
\newcommand\NA{\naturals^A}
\newcommand\DD{\mathcal{D}}
\renewcommand\SS{\mathcal{S}}
\newcommand\DNS{\DD(\NS)}

\newcommand\SNS{\SS(\NS)}

\newcommand\SNs{\SS(\Ns)}
\newcommand\SNt{\SS(\Nt)}

\newcommand\rest\restriction
\newcommand\dirac[1]{\delta_{#1}}
\newcommand\mset[1]{\{\kern-2pt|#1|\kern-2pt\}}
\newcommand\flatten{\flat}
\newcommand\struct\theta
\renewcommand\next[1]{\mathsf{next}\kern.5pt#1}

\newcommand\monom{\mathsf{monom}}
\newcommand\RRAlg{\reals\hyphen\mathsf{Alg}}
\newcommand\norm[1]{\|{#1}\|}
\renewcommand\angle[1]{\langle#1\rangle}
\newcommand\bone{\pmb{1}}
\newcommand\SR{(\Sigma^*,\reals)}
\newcommand\NX{\naturals^X}
\newcommand\DNX{\DD(\naturals^X)}
\newcommand\MRAlg{(M,R)\mathsf{\hyphen Alg}}
\newcommand\SRAlg{(\Sigma^*,\reals)\mathsf{\hyphen Alg}}
\newcommand\free{FS}

\newcommand\SeSSc{\Sigma^*\cd\beps,\Sigma^*\cd S}
\newcommand\dean{d_{\eps,a}^{\beps^n}}
\newcommand\fix[2]{\mathsf{fix}\kern1pt#1\kern1pt.\kern1pt#2}
\newcommand\KAT{\ensuremath{\mathsf{KAT}}}
\newcommand\GKAT{\ensuremath{\mathsf{GKAT}}}
\newcommand\PKA{\ensuremath{\mathsf{PKA}}}
\newcommand\cmp{\mathrel{;}}
\newcommand\Var{\mathsf{Var}}
\newcommand\semr[1]{\sem{#1}\rho}
\newcommand\bind{\mathrel{\texttt{>\kern-0.5pt>=}}}

\makeatletter
\newcommand*\cd{\mathpalette\bigcdot@{.6}}
\newcommand*\bigcdot@[2]{\mathbin{\vcenter{\hbox{\scalebox{#2}{$\m@th#1\bullet$}}}}}
\makeatother

\newcommand\extagqed{\tag*\qedhere}
\newcommand\exqed{}

\newcommand\seq[3]{#1_{#2},\ldots,#1_{#3}}
\newcommand\lam[2]{\lambda#1\kern1pt.\kern1pt#2}

\newcommand\naturals{{\mathbb N}}
\newcommand\reals{{\mathbb R}}

\newcommand\set[2]{\{#1 \mid #2\}}

\newcommand\subs{\mathrel{\subseteq}}

\newcommand\eps{\varepsilon}
\newcommand\Set{\ensuremath{\mathsf{Set}}}

\newcommand\len[1]{|#1|}

\newcommand\Iff{\Leftrightarrow}
\newcommand\skp{\mathtt{skip}} 
\newcommand\fail{\mathtt{fail}}
\newcommand\subst[2]{[#1/#2]}
\newcommand\hyphen{\mbox{-}}
\newcommand\del\partial
\newcommand\ite[3]{\mathtt{if}\ #1\ \mathtt{then}\ #2\ \mathtt{else}\ #3}
\newcommand\whd[2]{\mathtt{while}\ #1\ \mathtt{do}\ #2}
\newcommand\PSPACE{{\it PSPACE\/}}

\title{A Decision Procedure for Probabilistic Kleene Algebra\\with Angelic Nondeterminism}
\author[1]{Shawn Ong}
\author[2]{Dexter Kozen}
\affil[1]{Grinnell College}
\affil[2]{Cornell University}

\begin{document}

\maketitle

\begin{abstract}
In Ong et al.~(2025), a new finite automaton model with probabilistic and nondeterministic transitions was introduced, along with a corresponding regular expression language. The key difference from previous approaches is that nondeterminism is formalized in terms of multisets instead of powersets. Automata and expressions are interpreted over the space of probability distributions over multisets of strings with finite multiplicities. The main result of that paper was a full Kleene theorem showing that the two formalisms are equally expressive. In this paper we give a decision procedure for the semantic equivalence of two automata or expressions. The algorithm is based on polyomial algebra and makes use of Gr\"obner basis techniques for ideal membership testing.
\end{abstract}

\section{Introduction}
\label{sec:intro}

The combination of probability and nondeterminism in state-based systems is a notoriously thorny problem with a long and storied history \cite{Affeldt21,ChenSanders09,DahlqvistParlantSilva18,DArgenio98,FKMRS16,FTN08,Glabbeek90,GoyPetrisan20,Hansson94,HartogdeVink99,Hatefi17,Hermanns02,KeimelPlotkin17,LarsenSkou91,MciverCohenMorgan06,MGCM08,Mislove00,MisloveOuaknineWorrell04,MRS11,PnueliZuck93,QWGW08,Rennela16,Segala95,SegalaLynch94,SegalaPhD95,Sokolova11,Varacca03,VaraccaWinskel06,Vardi85,WangHoffmannReps19,ZKSTa24,Zwart20,ZwartMarsden22}.
As computational effects, the preferred formalization of nondeterminism has traditionally involved the powerset monad, whereas for probability, the Giry (distribution) monad is the formalization of choice. However, for two effects to coexist seamlessly in a single system, there should ideally be a \emph{distributive law} between them, analogous to the distributive law of multiplication over addition; but it has been known for some time that no such distributive law exists between the powerset and Giry monads \cite{VaraccaWinskel06,ZwartMarsden22,DahlqvistNeves18}. Many researchers over the last few decades have suggested solutions and workarounds, but the lack of distributivity renders these approaches less than ideal in various ways.

A relatively recent discovery, independently by Dash and Staton \cite{DashStaton21a,DashStaton21b,Dash23} and Jacobs \cite{Jacobs21}, is the existence of a distributive law between the \emph{multiset} monad and the Giry monad. This allows multisets of distributions to be transformed into distributions on multisets, just as products of sums can be transformed into sums of products. Based on this development, a new finite automaton model with probabilistic and nondeterministic transitions was introduced in \cite{OMK25a}, along with a corresponding regular expression language. The key difference from previous approaches is that the nondeterminism choice operator $\amp$ is formalized in terms of multisets instead of powersets. Automata and expressions are interpreted over $\DNS$, the space of probability distributions over multisets of strings with finite multiplicities. Thus strings are accepted not just with some probability, but with some finite multiplicity with some probability. The main result of \cite{OMK25a} was a full Kleene theorem showing that the two formalisms are semantically equivalent. This theorem was the first of its type among a long history of automata models incorporating probability and nondeterminism.

The main result of the present paper is that the semantic equivalence of two given automata or expressions is decidable. The algorithm is inspired by the decision procedure for weighted automata as presented for example in \cite{Kiefer20,Mohri09}. With weighted automata, the treatment depends on linear algebra. For a weighted automaton over $\reals$ with finite state set $S$, the algorithm relies on a dimension argument in the vector space $\reals^S$. A basis for a certain linear subspace of $\reals^S$ called the \emph{forward space} is computed inductively; with each step, either the dimension grows, or the algorithm terminates. Eventual termination is guaranteed by the finite dimensionality of $\reals^S$. For us, the structure of the algorithm is similar, except linear algebra is replaced by polynomial algebra. The appropriate mathematical structures are \emph{commutative algebras over $\reals$} ($\reals$-algebras). For us, the forward space is an ideal in the polynomial algebra $\reals[S]$. Termination for an automaton with finite state set $S$ is guaranteed by the Hilbert basis theorem, which implies that this algebra is \emph{Noetherian}; that is, there is no infinite proper ascending chain of ideals. The algorithm makes use of Gr\"obner basis techniques \cite{Buchberger76,Lazard83} for maintaining a basis of the forward space and for ideal membership testing.

The connection to polynomial algebra was not available to the authors of \cite{DashStaton21a,DashStaton21b,Dash23,Jacobs21,OMK25a}, and a portion of our development involves a reformulation of the results of those papers in terms of $\reals$-algebras, leading to some simplifications and new insights. We show that the space $\SNS$ of signed measures over multisets of strings with finite multiplicities is an $\reals$-algebra in which multiplication is nondeterministic choice $\amp$. The multiplicative unit is $\dirac{\mset{}}$, the Dirac (point mass) measure on the empty multiset. The space $\SNs$ of measures of finite support on finite multisets of a finite state set $S$ is also an $\reals$-algebra, and we show that it is naturally isomorphic to the polynomial algebra $\reals[S]$. A major epiphany for us was the realization that the distributivity of multisets over distributions \cite{DashStaton21a,DashStaton21b,Dash23,Jacobs21} is not only analogous to the distributivity of multiplication over addition, it \emph{is} the distributivity of multiplication over addition.

\subsection{Related Work}
\label{sec:related}

There is a long history of related work addressing the combination of probability and nondeterminism going back to \cite{Vardi85,LarsenSkou91}, including \cite{Affeldt21,ChenSanders09,DahlqvistParlantSilva18,DArgenio98,FKMRS16,FTN08,Glabbeek90,GoyPetrisan20,Hansson94,HartogdeVink99,Hatefi17,Hermanns02,KeimelPlotkin17,LarsenSkou91,MciverCohenMorgan06,MGCM08,Mislove00,MisloveOuaknineWorrell04,MRS11,PnueliZuck93,QWGW08,Rennela16,Segala95,SegalaLynch94,SegalaPhD95,Sokolova11,Varacca03,VaraccaWinskel06,Vardi85,WangHoffmannReps19,ZKSTa24,Zwart20,ZwartMarsden22}. A comprehensive comparison of the model of \cite{OMK25a} with previous approaches was given in \cite{OMK24a,OMK25a}.

In previous works, nondeterminism is most often regarded as \emph{demonic}, meaning that it arises from an adversarial external agent that is not under the control of the programmer. One of the original motivations \cite{Vardi85}, and still a major theme, is concurrency. Here the external agent is typically a scheduler that can interleave events in arbitrary ways, and it is up to the programmer/verifier to establish correctness regardless of the adversary's choices. On the other hand, in the approach of \cite{OMK25a} and the present paper, nondeterminism is viewed \emph{angelically}. There is a programming construct $\amp$ for nondeterministic choice that is under the control of the programmer. This construct enables KAT-style equational reasoning in the style of \cite{BeckettCampbellGreenberg17,AFGJKSW13a,SFHKKS20a,JFKKSSW25a}.

Due to the lack of a suitable distributive law between probability distributions and powersets, many previous models needed to restrict distributivity in some way or even disallow it outright. In addition, demonic nondeterminism often requires allowing the adversary extra nonconstructive power such as the ability to choose any convex combination of allowed distributions \cite{MciverCohenMorgan06,MGCM08,Segala95,ZKSTa24,MisloveOuaknineWorrell04}. Some approaches include the geometrically convex monad \cite{Affeldt21}, monad lifting \cite{DahlqvistParlantSilva18}, weak distributive laws \cite{GoyPetrisan20,Varacca03,VaraccaWinskel06}, and Kegelspitzen \cite{KeimelPlotkin17,Rennela16}. 

The angelic model of \cite{OMK25a} has a few other advantages as well. One advantage is that the use of multisets retains independence information that is lost with powersets, and this is perhaps the main reason for the lack of a distributive law with powersets. Another advantage is that the semantic domain $\DNS$ allows a finer comparison between models. Traditionally, two probabilistic automata are considered equivalent if each word is accepted by both automata with the same probability \cite{Rabin63,Tzeng92,Kiefer11}, which ignores correlations between the acceptance probabilities of words; whereas elements of $\DNS$ are joint distributions, so correlations are preserved.

With multisets, one can still ask for the probability that a string is accepted at all, that is, accepted with nonzero multiplicity, as this is a measurable observation.

\subsection{Roadmap}
\label{sec:roadmap}

In \S\ref{sec:prelim}, we review the automata and expressions introduced in \cite{OMK25a} and their semantics.

In \S\ref{sec:algebra} we show how to reformulate these constructs in terms of polynomial algebra and coalgebra. The two main results of this section show that the family of bounded, signed measures over $\NS$, the space of multisets of strings with finite multiplicities, is an $\reals$-algebra with multiplication interpreted as nondeterminstic choice, and that set of bounded signed measures over $\Ns$ for a finite set $S$ and finite support is naturally isomorphic to the polynomial ring $\reals[S]$.

In \cite{OMK25a}, the semantics of sequential composition was formulated in terms of an \emph{injective monoid action} of $\Sigma^*$ on the space of measures $\DNS$. This was done in the context of multisets. In \S\ref{sec:monoidactions}, we define \emph{$(M,R)$-algebras}, $R$-algebras with a specified injective monoid action of a monoid $M$, satisfying certain natural algebraic properties. We show in \S\ref{sec:coalgebras} how automata can be viewed as coalgebras $\struct:S\to\reals[\beps,\Sigma\cd S]$, where $\cd$ is the monoid action. This is a direct translation of the multiset formulation $\struct:S\to\DD(\naturals\times(\Ns)^\Sigma)$ of \cite{OMK25a} to the context of $(\Sigma^*,\reals)$-algebras.

In \S\ref{sec:technicalities} we set the stage for the correctness proof of our decision procedure by developing further the algebraic properties if $\SR$-algebras and automata. These include a notion of partial evaluation (\S\ref{sec:topoly}), how the semantic maps for automata and expressions translate to the context of $\SR$-algebras (\S\ref{sec:behavequations}), and an algebraic criterion for behavioral equivalence in terms of ideal membership (\S\ref{sec:behavequivalence}).

In \S\ref{sec:algorithm} we present our algorithm, along with a proof of correctness and termination. 
The algorithm is reminiscent of the algorithm for weighted automata \cite{Kiefer20,Mohri09} but works in the context of polynomial algebra instead of linear algebra. The algorithm uses the Buchberger algorithm for calculating a Gr\"obner basis of a polynomial ideal \cite{Buchberger76,Lazard83}.

Finally, in \S\ref{sec:conclusion} we draw conclusions and mention some open problems for future research.


\section{\PKA\ Automata and \PKA\ Expressions}
\label{sec:prelim}

In this section we briefly review \PKA\ automata and expressions and their semantics as presented in \cite{OMK25a}. In \S\ref{sec:algebra} we will reformulate much of this material in terms of $\reals$-algebras.

\subsection{\PKA\ Automata}
\label{sec:automata}

An automaton consists of a finite set of states of four types:
\begin{itemize}
 \item
 \emph{nondeterministic (choice) states} labeled $\amp$ with a finite multiset of successor states;
 \item
 \emph{probabilistic states} labeled $\oplus$ with a finite distribution on successor states;
 \item
 \emph{terminal states} labeled $\skp$ (accept) or $\fail$ (reject) with no successor states; and
 \item
 \emph{action states} labeled $a\in\Sigma$ with one successor state.
 \end{itemize}
 
Informally, the operation of an automaton can be described in terms of agents. At any time in the computation, there can be multiple agents, each occupying a state and acting independently of the other agents. A single agent originates at the start state. Thereafter, the computation proceeds as follows:
\begin{itemize}
\item
Each agent visiting a nondeterministic state $s$ labeled $\amp$ replicates itself and sends a copy to each state in the successor multiset, respecting multiplicities.
\item
Each agent visiting a probabilistic state $s$ labeled $\oplus$ independently samples the distribution associated with $s$, yielding a successor state to which the agent moves.
\item
An agent visiting a terminal state labeled $\skp$ halts and accepts the string of letters of $\Sigma$ it has seen so far. An agent visiting a terminal state labeled $\fail$ halts and rejects.
\item
An agent visiting an action state labeled $a$ with successor $t$ appends $a$ to the string of letters it has seen and moves to $t$.
\end{itemize}

The string $x$ is accepted with multiplicity $k$ if $k$ is the number of agents accepting that string. Of course, this occurs with some probability, depending on the probabilistic choices of the agents. Moreover, the probabilities for different strings and different multiplicities may be correlated. Thus the behavior of the automaton is best described by a joint distribution on the space of multisets of strings $\DNS$.

We impose a \emph{productivity assumption} that every cycle in the automaton must contain an action state. This assumption precludes infinite multiplicities and guarantees that fixpoints are unique, allowing a coalgebraic treatment of semantics.

Fig.~\ref{fig:coalg} in \S\ref{apx:automata} is an illustrative example of a fragment of an automaton.


Formally, an automaton is a tuple $(S,\Sigma,\ell,\struct)$, where $S$ is a finite set of states, $\Sigma$ is a finite alphabet of input letters, $\ell:S\to\{{\amp},\oplus,\skp,\fail\}\cup\Sigma$ is a labeling function, and $\struct:S\to \DD S + \Ns + S + 1$ are the transitions, such that:
\begin{itemize}
\item
if $\ell(s)=\oplus$, then $\struct(s)\in \DD S$, that is, $\struct(s)$ is a probability measure on $S$;
\item
if $\ell(s)={\amp}$, then $\struct(s)\in\Ns$, that is, $\struct(s)$ is a multiset of elements of $S$;
\item
if $\ell(s)\in\Sigma$, then $\struct(s)\in S$, that is, $s$ has one successor; and
\item
if $\ell(s)\in\{\skp,\fail\}$, then $s$ has no successors.
\end{itemize}

Semantically, every state in $S$ represents a distribution over multisets of strings in $\DNS$. Let $\dirac m$ denote the Dirac (point mass) measure on the multiset $m$.
The semantic map $\sem-:S\to\DNS$ is defined coinductively:
\begin{align*}
\sem s &= \begin{cases} 
\dirac{\mset\eps}, & \text{if $\ell(s)=\skp$}\\
\dirac{\mset{}}, & \text{if $\ell(s)=\fail$}\\
a \cd\sem t, & \text{if $\ell(s)= a \in\Sigma$ and $\struct(s)=t$}\\
\sum_i r_i\sem{t_i}, & \text{if $\ell(s)=\oplus$ and $\struct(s)=\sum_i r_it_i$}\\
{\amp}(M\sem-(m)), & \text{if $\ell(s)={\amp}$ and $\struct(s)=m$.}
\end{cases}
\end{align*}
The injective monoid action $(\cd)$ of $\Sigma^*$ on $\DNS$ appearing in the third case plays a prominent role, and we will discuss it further in \S\ref{sec:monoidactions}.
In the last case, $M\sem-$ is the map that, given a multiset $m$, applies $\sem-$ to every element of $m$ and takes the multiset of results, and $\amp$ is the semantic operation on multisets of measures whose binary version is
\begin{align}
& {\amp}:\DNS^2\to\DNS && \mu\amp\nu = (\mu\otimes\nu)\circ\uplus^{-1}.\label{eq:nondetasmult}
\end{align}
where $\otimes$ is the product distribution operator and $\uplus$ is multiset union (pointwise sum).
This is the pushforward measure of the product distribution $\mu\otimes\nu:\DNS^2$ under multiset union $\uplus$, a measurable function $(\NS)^2\to\NS$. Operationally, given distributions $\mu$ and $\nu$, sample $\mu$ and $\nu$ independently to obtain two multisets of strings, then take their multiset union. The higher arity version used in the semantic definition above, when given a multiset $m$ of states, would
\begin{enumerate}
\item
apply $M\sem-$ to $m$, which interprets all elements of $m$ by $\sem-$ to obtain a multiset of distributions in $\DNS$,
\item
sample them all independently to obtain a multiset of outcomes, which are multisets of strings in $\NS$,
\item
take the multiset union of the outcomes. 
\end{enumerate}
For the special case of binary $\amp$, this reduces to
\begin{align*}
\sem{s\amp t} &= \sem s\amp\sem t = (\sem s\otimes\sem t)\circ{\uplus}^{-1}.
\end{align*}
Similarly, for the special case of binary $\oplus_r$,
\begin{align*}
\sem{s\oplus_r t} &= r\sem s + (1-r)\sem t.
\end{align*}

It is shown in \cite{OMK25a} that the map $\sem-$ is well defined. Briefly, the productivity assumption ensures that a map modeling the coinductive definition is contractive in a certain complete metric space, thus by the Banach fixpoint theorem has a unique fixpoint $\sem-$.

It is argued in \cite{OMK25a} that every automaton can be transformed to a canonical coalgebraic form $\struct:S\to\DD(\naturals\times(\Ns)^\Sigma)$. We will say more about this representation in \S\ref{sec:coalgebras}.

\subsection{\PKA\ Expressions}
\label{sec:expressions}

Let $\Var=\{x,y,\ldots\}$ be a set of variables and $\Sigma = \{a,b,\ldots\}$ a finite set of letters disjoint from $\Var$. The language of expressions is given by the BNF grammar
\begin{align*}
e\ &::=\ x \mid a \mid e_1 \amp e_2 \mid e_1 \oplus_r e_2 \mid e_1\cmp e_2 \mid \fix xe \mid \skp \mid \fail
\end{align*}
The operator $\cmp$ is for sequential composition. In any expression of the form $e_1\cmp e_2$, it is required that the left operand $e_1$ be closed (contain no free variables), a restriction that characterizes linear recursion. Thus expressions are the appropriate analog of regular (rational) expressions in this context.

It is also required that in fixpoint expressions $\fix xe$, all paths from the root of $e$ to a free occurrence of $x$ must pass through some $a\in\Sigma$. This is another manifestation of the \emph{productivity assumption}.


Like automata, expressions $e$ are interpreted as measures on multisets of strings $\semr e:\DNS$ relative to an environment $\rho:\Var\to\DNS$. The environment $\rho$ is used to interpret free variables. A closed expression (one in which all variables are bound by $\mathsf{fix}$) does not depend on $\rho$.
\begin{gather*}
\semr\skp = \dirac{\mset\eps} \qquad
\semr\fail = \dirac{\mset{}} \qquad
\semr a = \dirac{\mset a}\\
\semr x = \rho(x) \qquad
\semr{e_1 \oplus_r e_2} = r\semr{e_1} + (1-r)\semr{e_2}\\
\semr{e_1 \amp e_2} = \semr{e_1}\amp\semr{e_2} = (\semr{e_1}\otimes\semr{e_2})\circ{\uplus}^{-1}\\
\semr{e_1 \cmp e_2} = \sem{e_1} \bind -\cd\semr{e_2}\\
\semr{\fix xe} = \semr e[\semr{\fix xe}/x]
\end{gather*}
The notations $\dirac m$, $\otimes$, $\uplus$, and $\amp$ are as in the previous section. The definitions for $e_1 \cmp e_2$ and $\fix xe$ contain some undefined notation, which we will explain now.

The semantics of sequential composition $\cmp$ is based on Kleisli composition involving a bind operation ($\bind$)
introduced in \cite{OMK25a}. As we do not need the bind operation itself for our development, we do not elaborate on the formal definition here, except to say that it is defined in terms of the same monoid action that appeared in the definition of semantics for automata (and which will recur in \S\ref{sec:monoidactions}).

Operationally, to sample $\semr{e_1\cmp e_2}$, we sample $\sem{e_1}$ to obtain a multiset $\beta$, then sample $x\cd\semr{e_2} = \semr{x\cmp e_2}$ for each $x\in\beta$ and take their multiset union. We have required that the left operand $e_1$ be closed in compositions $e_1\cmp e_2$. This restriction precludes nonlinear behavior. It also means that the environment $\rho$ is not needed to evaluate $\sem{e_1}$.

In the definition of $\semr{\fix xe}$, the notation $\rho\subst\mu x$ refers to the environment $\rho$ with the variable $x$ rebound to $\mu$. The definition may appear circular, but it is a valid coinductive definition due to the productivity assumption. It is shown in \cite{OMK25a} that the semantics is well defined.

The expression $\fix xe$ represents the unique solution of the equation $x=e$
in $\DNS$, where $e$ may contain free occurrences of $x$. A special case is the traditional $^*$ operator of Kleene algebra
\begin{align*}
e^* &= \fix x{\skp\amp(e\cmp x)},
\end{align*}
which is the unique solution of the equation $x = \skp\amp(e\cmp x)$
in $\DNS$. Thus we have valid equations
\begin{align*}
\fix xe &= e[(\fix xe)/x] & e^* &= \skp\amp (e\cmp e^*)
\end{align*}
as shown formally in \cite{OMK25a}.

\section{Automata as Polynomial Coalgebras}
\label{sec:algebra}

In this section we show how constructs of the previous section can be reformulated in terms of polynomial algebra and coalgebra. The main results of this section are Theorems \ref{thm:RAlgSigma} and \ref{thm:RAlgS}, which show that the distributivity of multisets over distributions \cite{DashStaton21a,DashStaton21b,Dash23,Jacobs21}
 is nothing more than the distributivity of multiplication over addition.

\subsection{$R$-algebras}
\label{sec:RAlg}

Let $R$ be a unital commutative ring. A \emph{commutative algebra over $R$} ($R$-algebra) is an $R$-module $V$ with a unital commutative bilinear multiplication. The class of $R$-algebras and structure-preserving maps form a category $\RAlg$. Two other equivalent equational presentations are 
\begin{itemize}
\item
unital commutative rings $V$ with a ring action $R\times V\to V$ as scalar multiplication;
\item
unital commutative rings $V$ with a unital commutative ring homomorphism $R\to V$.
\end{itemize}
The relation between the latter two characterizations is as follows. Given a scalar multiplication $(\cdot):R\times V\to V$, the corresponding ring homomorphism $R\to V$ is $r\mapsto r1_V$; conversely, given a ring homomorphism $h:R\to V$, the corresponding scalar multiplication $R\times V\to V$ is $(r,a) \mapsto h(r)a$.

The multivariate polynomial ring $R[S]$ is the free $R$-algebra on generators $S$. The free construction $R[-]$, being left adjoint to a forgetful functor, gives rise to a monad over \Set. The monad multiplication $R[R[S]]\to R[S]$ regards a polynomial of polynomials over $S$ as a polynomial over $S$. If desired, one can multiply out the nested polynomial expressions and combine like terms. Any polynomial expression can be reduced to a canonical representation as an $R$-weighted sum of distinct monomials in this way. The monad unit $S\to R[S]$ maps $s$ as an indeterminate to $s$ as a monomial. The category $\RAlg$ of $R$-algebras is the category of Eilenberg-Moore algebras of this monad.

Let $(V,\sigma)$ be an $R$-algebra regarded as an Eilenberg-Moore algebra for the monad $R[-]$, where $\sigma:R[V]\to V$ is the structure map.
Any interpretation $f:S\to V$ of the indeterminates extends uniquely to an evaluation morphism $f^\dagger = \sigma\circ R[f] : R[S]\to V$, where $R[f]:R[S]\to R[V]$ is the morphism that substitutes $f(s)$ for $s$.

Let $(V,\sigma)$ be an $R$-algebra. An \emph{ideal} in $V$ is a subset $I\subs V$ such that
\begin{itemize}
\item
$0\in I$
\item
if $x,y\in I$, then $x+y\in I$
\item
if $x\in V$ and $y\in I$, then $xy\in I$.
\item
if $r\in R$ and $y\in I$, then $ry\in I$.
\end{itemize}
The \emph{kernel} of an $\RAlg$ homomorphism $h:V\to W$ is the set $\ker h = \set{a\in V}{h(a)=0}$. It is easily shown that the kernel of any $\RAlg$ homomorphism is an ideal. Conversely, any ideal in $V$ is the kernel of an epimorphism obtained via a quotient construction.

\subsection{Signed Measures on Multisets as $\reals$-algebras}
\label{sec:signed}

Let $\Sigma$ be a finite alphabet. Let $\Sigma^*$ denote the set of finite-length strings over $\Sigma$ and $\eps$ the null string. The family of multisets of $\Sigma^*$ with finite multiplicities is $\NS$, the family of functions $\Sigma^*\to\naturals$. Multiset union is just pointwise addition and is denoted $\uplus$. Multiset comprehension is denoted with stylized braces $\mset{{-}}$.

A \emph{finite tree} is a finite subset of $\Sigma^*$ closed under the prefix relation. For any finite tree $A$, let $\del A$ be the set of prefix-minimal elements of $\Sigma^*\setminus A$. For example, if $A = \{\eps,a,b,aa,ab\}$, then $\del A = \{ba,bb,aaa,aab,aba,abb\}$. Also, $\del\emptyset = \{\eps\}$, $\del\{\eps\} = \Sigma$, and $\del\Sigma^{<n} = \Sigma^n$, where $\Sigma^{<n} = \set{x\in\Sigma^*}{\len x<n}$ and $\Sigma^{n} = \set{x\in\Sigma^*}{\len x=n}$.

As argued in \cite{OMK25a}, $\NS$ is a standard Borel space with measurable sets generated by the sets $[\alpha]_A = \set{\beta\in\NS}{\beta\equiv_A\alpha}$, where $A\subs\Sigma^*$ is a finite tree, $\alpha\in\NS$, and $\beta\equiv_A\alpha$ if $\beta(x)=\alpha(x)$ for all $x\in A$.
Each $[\alpha]_A$ contains a unique canonical element $\alpha\rest A$ that agrees with $\alpha$ on $A$ and vanishes outside of $A$.
Since the support of $\alpha\rest A$ is contained in $A$, we can view $\alpha\rest A$ as an element of $\NA$.

The space $\DNS$ of probability measures on $\NS$ is the space of behaviors over which the automata and expressions of \cite{OMK25a} are interpreted. Let $\SNS$ denote the space of bounded signed measures on $\NS$. Every element of $\SNS$ can be represented as a weighted difference $a\mu-b\nu$ for some $\mu,\nu\in\DNS$ and $a,b\in\reals$. The nondeterministic choice operator $\amp$ defined in \eqref{eq:nondetasmult} can be extended to $\SNS$ by linearity:
\begin{align*}
& {\amp}:\SNS^2\to\SNS && \mu\amp\nu = (\mu\otimes\nu)\circ\uplus^{-1}.
\end{align*}

\begin{theorem}
\label{thm:RAlgSigma}
The structure $(\SNS,+,0,\amp,\dirac{\mset{}})$ is an $\reals$-algebra, with nondeterministic choice $\amp$ as multiplication and $\dirac{\mset{}}$ as multiplicative unit.
\end{theorem}
\begin{proof}
It is clearly a vector space over $\reals$ under the pointwise operations $+$, $-$, and $0$. Multiplication is nondeterministic choice $\amp$. This is defined by $\mu\amp\nu = (\mu\otimes\nu)\circ\uplus^{-1}$, the pushforward measure of the independent product $\mu\otimes\nu$ on $(\NS)^2$ under multiset union $\uplus$. Intuitively, to sample $\mu\amp\nu$, one samples $\mu$ and $\nu$ independently, then takes the multiset union of the outcomes. This operation is associative and commutative because the independent product operation $\otimes$ is associative and commutative on distributions and multiset union $\uplus$ is associative and commutative on multisets. The unit for $\amp$ is $\dirac{\mset{}}$, the Dirac (point mass) measure on the empty multiset, so the structure is a commutative monoid under $\amp$ and $\dirac{\mset{}}$. Finally, distributivity of $\amp$ over weighted addition in $\DNS$ (that is, $\mu\amp(r\nu+(1-r)\rho) = r(\mu\amp\nu)+(1-r)(\mu\amp\rho)$) was shown in \cite[Lemma B.5]{OMK24a}. This is easily extended to full distributivity of $\amp$ over $+$ on $\SNS$ by linearity, where it takes the form
\begin{align*}
\mu\amp(a\nu+b\rho) = a(\mu\amp\nu)+b(\mu\amp\rho).
\extagqed
\end{align*}
\end{proof}

The same structure applies to automata. Let $S$ be a finite set. Let $\SNs$ denote the family of finite signed measures over $\Ns$ with finite support, with the operations defined in the same way as in $\SNS$.
\begin{theorem}
\label{thm:RAlgS}
The structure $(\SNs,+,0,\amp,\dirac{\mset{}})$ is an $\reals$-algebra and is naturally isomorphic to the polynomial algebra $\reals[S]$.
\end{theorem}
\begin{proof}
The proof that $\SNs$ is an $\reals$-algebra is the same as in Theorem \ref{thm:RAlgSigma} except for the distributivity argument. However, here the argument is easier, since the space is discrete. To show distributivity, for any measurable $B\subs\Ns$,
\begin{align*}
(\mu \amp \nu)(B)
&= ((\mu \otimes \nu)\circ\uplus^{-1})(B)\\
&= (\mu \otimes \nu)(\set{(m,\ell)}{m\uplus\ell\in B})
= \sum_{m\uplus\ell\in B}\mu(m)\cdot\nu(\ell),
\end{align*}
therefore
\begin{align*}
((\mu \amp (a\nu + b\rho))(B)
&= \sum_{m\uplus\ell\in B}\mu(m)\cdot(a\nu + b\rho)(\ell)\\
&= a\sum_{m\uplus\ell\in B}\mu(m)\cdot \nu(\ell) + b\sum_{m\uplus\ell\in B}\mu(m)\cdot \rho(\ell)\\
&= a(\mu\amp\nu)(B) + b(\mu\amp\rho)(B)\\
&= (a(\mu\amp\nu) + b(\mu\amp\rho))(B).
\end{align*}
As $B$ was arbitrary, $\mu \amp (a\nu + b\rho) = a(\mu\amp\nu) + b(\mu\amp\rho)$. Thus the structure is an $\reals$-algebra.

Now we argue that $\SNs$ is naturally isomorphic to $\reals[S]$. Since $\reals[S]$ is the free $\reals$-algebra on generators $S$, there is a unique $\reals$-algebra homomorphism $\phi_S:\reals[S]\to\SNs$ extending $s\mapsto\dirac{\{s\}}$. This map specializes to a bijection $\phi_S:\monom[S]\to\naturals^S$, where $\monom[S]$ denotes the set of monomials of $R[S]$, in which $\phi_S(m)(s)$ is the degree of $s$ in $m$. For example, $\phi_S(s^2t^3) = \mset{s,s,t,t,t}$. Moreover, this bijection is an isomorphism of commutative monoids $(\monom[S],\cdot,1)\to(\Ns,\uplus,\mset{})$, and both structures are representations of the free commutative monoid on generators $S$. As elements of $\reals[S]$ and $\SNs$ are simply finite linear combinations of $\monom[S]$ and $\Ns$, respectively, $\phi_S:\reals[S]\to\SNs$ is an isomorphism of $\reals$-algebras.

The isomorphism is natural in the sense that for any set function $f:S\to T$, the diagram
\begin{align*}
\begin{tikzpicture}[->, >=stealth', auto]
\small
\node (NW) {$\reals[S]$};
\node (NE) [right of=NW, node distance=22mm] {$\SNs$};
\node (SW) [below of=NW, node distance=11mm] {$\reals[T]$};
\node (SE) [below of=NE, node distance=11mm] {$\SNt$};
\path (NW) edge node[swap] {$\reals[f]$} (SW);
\path (NW) edge node {$\phi_S$} (NE);
\path (NE) edge node {$\SS(\naturals^f)$} (SE);
\path (SW) edge node[swap] {$\phi_T$} (SE);
\end{tikzpicture}
\end{align*}
commutes, where $\SS(\naturals^f)$ replaces each occurrence of $s\in S$ with $f(s)$ in a given multiset and $\reals[f]$ substitutes $f(s)$ for the indeterminate $s$ in a given polynomial.
\exqed
\end{proof}

\subsection{Monoid Actions and $(M,R)$-Algebras}
\label{sec:monoidactions}

In \cite{OMK25a}, the notion of an \emph{injective monoid action} $(\cd):M\times X\to X$ of a monoid $M$ on a set $X$ was used to define the semantics of sequential composition. The action satisfies the rules $1\cd u=u$ and $xy\cd u = x\cd(y\cd u)$ and is injective in the sense that $x\cd u=x\cd v$ implies $u=v$. It was shown in \cite{OMK25a} that there is a canonical way to lift the monoid action to multisets $\NX$ and to distributions $\DNX$; moreover, the lifted actions are also injective.

In the application of \cite{OMK25a}, the monoid $M$ is the free monoid $\Sigma^*$, which acts injectively on $\Sigma^*$, $\NS$, and $\DNS$ in the following way: for $x,y\in\Sigma^*$, $\alpha\in \NS$, and $\mu\in \DNS$,
\begin{align}
x\cd y &= xy\nonumber\\
x\cd\alpha &= \mset{x\cd z \mid z\in \alpha} = \lam y{\begin{cases}
\alpha(z), & \text{if $\exists z\ y=xz$},\\
0, & \text{otherwise}
\end{cases}} \nonumber\\
x\cd\mu &= \mu\circ(x\cd-)^{-1}. \label{eq:Saction} 
\end{align}

This construction inspires an extension of $R$-algebras to include a monoid action respecting the $R$-algebra structure. Let $V$ be an $R$-algebra with a distinguished element $\beps$ and let $M$ be a monoid. Let us define an \emph{$M$-action} on $V$ to be a map $(\cd):M\times V\to V$ such that
\begin{enumerate}[(i)]
\item
$(\cd)$ is a monoid action; that is,
\begin{align*}
x\cd(y\cd u) &= xy\cd u & 1\cd u &= u,
\end{align*}
\item
each $x\cd-:V\to V$ is an $R$-algebra homomorphism; that is,
\begin{gather*}
x\cd(ru+sv) = r(x\cd u)+s(x\cd v)\\
x\cd(uv) = (x\cd u)(x\cd v)\\
x\cd 1 = 1\qquad x\cd 0 = 0. 
\end{gather*}
\end{enumerate}
In addition, the monoid action is \emph{injective} provided
\begin{enumerate}[(i)]
\setcounter{enumi}2
\item
if $x\cd u = 0$, then $u=0$.
\end{enumerate}
Condition (ii) implies that the action $x\cd-$ on $V$ is uniquely determined by its action on a generating set. Condition (iii) implies that $x\cd-$ is a monomorphism.\footnote{A caveat: the monoid action $(\cd)$ should not be confused with multiplication $(\cdot)$ in $V$. Note that $\beps\ne 1$ and $x\cd\beps\ne\beps$ in general, whereas $x\cd 1=1$. The expression $x\cdot\beps$ does not make sense, as $x$ is not an element of $V$. We will typically use juxtaposition for multiplication in $M$.}

For a fixed monoid $M$, the $R$-algebras with an $M$-action are called \emph{$(M,R)$-algebras}. Morphisms of $(M,R)$-algebras are morphisms of $R$-algebras such that $f(x\cd u) = x\cd f(u)$ and $f(\beps) = \beps$. This category is denoted $\MRAlg$.

Here are some examples of $\SR$-algebras with injective $\Sigma^*$-actions that will play a role in our development:
\begin{itemize}
\item
$\SNS$ with the $\Sigma^*$-action \eqref{eq:Saction} for $\DNS$, extended to $\SNS$ by linearity. The distinguished element $\beps$ is $\dirac{\{\eps\}}$.
\item
$\free = \reals[\SeSSc]$, the algebra of polynomials with indeterminates $x\cd\beps$ and $x\cd s$ for $x\in\Sigma^*$ and $s\in S$. This is the free $\SR$-algebra generated by the set $S$. 
Here $\beps$ is simply an indeterminate. In our application, $S$ will be finite and will be the state set of an automaton.
\item
$\reals$ with $x\cd r = r$ and $\beps=1$.
\item
$V[\Sigma^*\cd S]$, the polynomial $\SR$-algebra obtained by freely adjoining a new set of indeterminates $x\cd s$ for $x\in\Sigma^*$ and $s\in S$ to a $\SR$-algebra $V$.
\end{itemize}

In the second example above we have a \Set-endofunctor $F = \reals[\Sigma^*\cd\beps,\Sigma^*\cd -]$. The action of $F$ on a function $f:S\to T$ substitutes $f(s)$ for $s\in S$ in a given polynomial expression:
\begin{align*}
& Ff:\reals[\SeSSc]\to\reals[\Sigma^*\cd\beps,\Sigma^*\cd T]
&& Ff(p) = p\subst{f(s)}{s \mid s\in S}.
\end{align*}
In $\free$, any expression over the $\reals$-algebra operators and $(\cd)$ can be flattened to canonical form $\reals[\SeSSc]$ by applying the rules (i) and (ii) in the left-to-right direction and combining like terms of the resulting polynomial. In particular, $F$ carries a monad structure whose multiplication $\flatten_S:F^2S\to FS$ is achieved in this way. The monad unit $S\to FS$ maps $s$ as an indeterminate to $s$ as a monomial. 

Since $\free$ is free on generators $S$, any set function $f:S\to V$, where $V$ is a $\SR$-algebra, lifts uniquely to a $\SRAlg$ morphism $f^\ddag:\free\to V$ using the properties (i) and (ii) above. More precisely, $f^\ddag = \sigma\circ Ff$, where $Ff:FS\to FV$ and $\sigma:FV\to V$ is the structure map of $V$, regarded as an Eilenberg-Moore algebra for the monad $F$.
The operation $^\ddag$ is just Kleisli lifting. In particular, if $f:S\to FT$, then $f^\ddag = \flatten_T\circ Ff : FS\to FT$, where $Ff:FS\to F^2T$ and $\flatten_T:F^2T\to FT$ is the monad multiplication. 
These operations compose as follows:
\begin{lemma}
\label{lem:prodf}
Let $S,T$ be sets and $V$ a $\SR$-algebra with structure map $\sigma$.
If $f:S\to FT$ and $g:T\to V$, then $(g^\ddag \circ f)^\ddag = g^\ddag \circ f^\ddag : FS\to V$.
\end{lemma}

\begin{proof}
Unwinding definitions, we have
\begin{align*}
(g^\ddag\circ f)^\ddag &= \sigma\circ F(g^\ddag\circ f) = \sigma\circ F(\sigma\circ Fg\circ f) = \sigma\circ F\sigma\circ F^2g\circ Ff\\
g^\ddag\circ f^\ddag &= \sigma\circ Fg\circ\flatten_T\circ Ff,
\end{align*}
so it suffices to show
$\sigma\circ F\sigma\circ F^2g = \sigma\circ Fg\circ\flatten_T$.
This is established by the following diagram.
\begin{align*}
\begin{tikzpicture}[->, >=stealth', auto]
\small
\node (NW) {$F^2T$};
\node (N) [right of=NW, node distance=22mm] {$F^2V$};
\node (NE) [right of=N, node distance=22mm] {$FV$};
\node (SW) [below of=NW, node distance=11mm] {$FT$};
\node (S) [below of=N, node distance=11mm] {$FV$};
\node (SE) [below of=NE, node distance=11mm] {$V$};
\path (NW) edge node {$F^2g$} (N);
\path (N) edge node {$F\sigma$} (NE);
\path (SW) edge node[swap] {$Fg$} (S);
\path (S) edge node[swap] {$\sigma$} (SE);
\path (NW) edge node[swap] {$\flatten_T$} (SW);
\path (N) edge node {$\flatten_V$} (S);
\path (NE) edge node {$\sigma$} (SE);
\end{tikzpicture}
\end{align*}
The left rectangle commutes because $\flatten$ is a natural transformation, and the right rectangle commutes due to a property of Eilenberg-Moore algebras for a monad.
\exqed
\end{proof}

\subsection{Automata as $\reals[\beps,\Sigma\cd S]$-coalgebras}
\label{sec:coalgebras}

In \cite{OMK25a}, it was shown that an automaton could be regarded as a coalgebra $(S,\struct)$ of type $\struct:S\to\DD(\naturals\times(\Ns)^\Sigma)$. Intuitively, an agent at state $s\in S$ samples the distribution $\struct(s)$, yielding an outcome in $\naturals\times(\Ns)^\Sigma$, then spawns several other agents, each of which either halts and accepts, halts and rejects, or scans a letter and proceeds to a new state. The outcome determines how many of each action occur.

Under the natural isomorphism $\phi_S$, the structure map translates to type $\struct:S\to\reals[\beps,\Sigma\cd S]$, where $\Sigma\cd S = \set{a\cd s}{a\in\Sigma,\ s\in S}$. The coefficient of $\beps^{k_\beps}\prod_{a\in\Sigma,\ s\in S}(a\cd s)^{k_{a\cd s}}$ is the probability that $k_\beps$ agents halt and accept and that $k_{a\cd s}$ agents scan $a\in\Sigma$ and move to state $s$. The constant coefficient is the probability of failure. For example, for the automaton fragment shown in Fig.~\ref{fig:coalg} in \S\ref{apx:automata},
\begin{align*}
\struct(s) &= p\beps^2(a\cd s)(b\cd t) + q\beps(a\cd u)(a\cd v) + r\eps^2(a\cd s)(b\cd t)^2
+ (1-p-q-r)(a\cd s)^2(a\cd t)(b\cd t).
\end{align*}
%
%
We can view $a\cd s$ and $\Sigma\cd S$ as just alternative notation for $(a,s)$ and $\Sigma\times S$, respectively, but the interpretation of $(\cd)$ as a monoid action will become significant when we discuss extensions to the structure map $\struct$ in \S\ref{sec:extension}.

\section{Properties of $\SR$-Algebras}
\label{sec:technicalities}

\subsection{Partial Evaluation}
\label{sec:topoly}

Let $V$ be a $\SR$-algebra and let $h:\Sigma^* \cd S\to V[\Sigma^* \cd S]$ be a function. Then $h$ extends uniquely to an $\reals$-algebra homomorphism $V[\Sigma^* \cd S]\to V[\Sigma^* \cd S]$ fixing $V$, which we also denote by $h$. Let $B\subs\Sigma^*$ be a set of pairwise prefix-incomparable strings. Let $B\cd h:\Sigma^* \cd S \to V[\Sigma^* \cd S]$ be the function
\begin{align*}
& (B\cd h)(w\cd s) = \begin{cases}
x\cd h(y\cd s) & \text{if $w = xy$ and $x\in B$}\\
w\cd s & \text{otherwise.}
\end{cases}
\end{align*}
This is well defined, because $w\in\Sigma^*$ can have at most one prefix $x\in B$. The unique $\reals$-algebra homomorphism $B\cd h:V[\Sigma^* \cd S]\to V[\Sigma^* \cd S]$ that extends this map and fixes $V$ simultaneously substitutes $x\cd h(y\cd s)$ for $xy\cd s$ in its argument for all $x\in B$. Let $x\cd h = \{x\}\cd h$.


We remark that $x\cd h$ is not necessarily a $\SR$-algebra homomorphism. For example, for $V[\Sigma^* \cd S]=\free$, if $a\cd h$ were a $\SR$-algebra morphism, then $a\cd h(s) = (a\cd h)(a\cd s) = a\cd((a\cd h)(s)) = a\cd s$. By injectivity of $a\cd-$, this only holds when $h$ is the identity function.

The following lemma establishes some properties of $x\cd h$ and $B\cd h$.
\begin{lemma}
\label{lem:cdprops}
Let $V$ be a $\SR$-algebra.
\begin{enumerate}[{\upshape(i)}]
\item
Let $h:\Sigma^* \cd S\to V[\Sigma^* \cd S]$. If each of $A$ and $B$ is a set of pairwise prefix-incomparable strings, then so is $AB$, and $AB\cd h = A\cd(B\cd h)$.
\item
Let $f,g:\Sigma^* \cd S\to V[\Sigma^* \cd S]$. Then $(x\cd g)\circ(x\cd f) = x\cd(g\circ f)$.
\item
Let $f:\Sigma^* \cd S\to V[\Sigma^* \cd S]$ such that $f(x\cd s) = x\cd f(s)$ (thus the $\reals$-algebra homomorphism fixing $V$ generated by $f$ is a $\SR$-algebra homomorphism). Let $A$ be a finite tree. Then $f$ and $\del A\cd f$ agree on $V[(\Sigma^*\setminus A)\cd S]$.
\item
Let $f,g:\Sigma^* \cd S\to V[\Sigma^* \cd S]$. Let $A$ and $B$ be disjoint sets of strings such that the strings in $A\cup B$ are pairwise prefix-incomparable. Then $(A\cd f)\circ(B\cd g) = (B\cd g)\circ(A\cd f)$. If $f=g$, this is just $(A\cup B)\cd f$.
\end{enumerate}
\end{lemma}

\begin{proof}
For (i), that the elements of $AB$ are pairwise prefix-incomparable is obvious. For the second statement, for any $w\cd s\in\Sigma^*\cd S$,
\begin{align*}
& (A\cd(B\cd h))(w\cd s)\\
&= \begin{cases}
x\cd (B\cd h)(y\cd s) & \text{if $w = xy$ and $x\in A$}\\
w\cd s & \text{otherwise}
\end{cases}\\
&= \begin{cases}
x\cd u\cd h(v\cd s) & \text{if $w=xy$, $x\in A$, $y = uv$, and $u\in B$}\\
x\cd y\cd s & \text{if $w=xy$ and $x\in A$, but $y$ has no prefix in $B$}\\
w\cd s & \text{otherwise}
\end{cases}\\
&= \begin{cases}
xu\cd h(v\cd s) & \text{if $w=xuv$, $x\in A$, and $u\in B$}\\
xy\cd s & \text{if $w=xy$ and $x\in A$, but $y$ has no prefix in $B$}\\
w\cd s & \text{otherwise}
\end{cases}\\
&= \begin{cases}
z\cd h(v\cd s) & \text{if $w=zv$ and $z\in AB$}\\
w\cd s & \text{otherwise}
\end{cases}\\
&= (AB\cd h)(w\cd s).
\end{align*}

For (ii), we have $(x\cd g)(x\cd y\cd s) = x\cd g(y\cd s)$ for every $y\cd s$. Since the two $\reals$-algebra homomorphisms $(x\cd g)\circ(x\cd -)$ and $(x\cd -)\circ g$ agree on generators $y\cd s$, they agree on all $p\in V[\Sigma^*\cd S]$. Then for every $w\cd s$,
\begin{align*}
(x\cd g)((x\cd f)(w\cd s))
&= (x\cd g)\left(\begin{cases}
x\cd f(y\cd s), & \text{if $w=xy$,}\\
w\cd s, & \text{otherwise}
\end{cases}\right)\\
&= \begin{cases}
(x\cd g)(x\cd f(y\cd s)), & \text{if $w=xy$,}\\
(x\cd g)(w\cd s), & \text{otherwise}
\end{cases}\\
&= \begin{cases}
x\cd g(f(y\cd s)), & \text{if $w=xy$,}\\
w\cd s, & \text{otherwise}
\end{cases}\\
&= (x\cd(g\circ f))(w\cd s).
\end{align*}

For (iii), it suffices to show that that $f$ and $\del A\cd f$ agree on the generators of $V[(\Sigma^*\setminus A)\cd S]$. For $w\cd s\in(\Sigma^*\setminus A)\cd S$,
%
%
\begin{align*}
(\del A\cd f)(w\cd s)
&= \begin{cases}
x\cd f(y\cd s), & \text{if $w=xy$ and $x\in\del A$}\\
w\cd s, & \text{otherwise}
\end{cases}\\
&= \begin{cases}
f(xy\cd s), & \text{if $w=xy$ and $x\in\del A$}\\
w\cd s, & \text{otherwise.}
\end{cases}
\end{align*}
But the second alternative cannot happen, since every element of $\Sigma^*\setminus A$ has a prefix in $\del A$.

For (iv), for every $p\in V[\Sigma^*\cd S]$ and $x\in B$, every indeterminate of $x\cd p$ is of the form $xz\cd s$, and $xz$ cannot have a prefix in $A$, thus $(A\cd f)(xz\cd s) = xz\cd s$, so $(A\cd f)(x\cd p) = x\cd p$. Then for any $w\cd s\in\Sigma^*\cd S$,
\begin{align*}
(A\cd f)((B\cd g)(w\cd s))
&= 
\begin{cases}
(A\cd f)(x\cd g(y\cd s)) & \text{if $w=xy$ and $x\in B$}\\
(A\cd f)(w\cd s) & \text{otherwise}
\end{cases}\\
&= 
\begin{cases}
x\cd g(y\cd s) & \text{if $w=xy$ and $x\in B$}\\
x\cd f(y\cd s) & \text{if $w=xy$ and $x\in A$}\\
w\cd s & \text{otherwise.}
\end{cases}
\end{align*}
At most one of the first two alternatives can occur by the assumption of prefix-incomparability. By symmetry, expanding $(B\cd g)((A\cd f)(w\cd s))$ would give the same value.
\exqed
\end{proof}

\subsection{Behavioral Equations}
\label{sec:behavequations}

The coinductive definition of the semantic map $\sem-:S\to\DNS$ from \cite{OMK25a} takes the following form in the present context.

In \cite{OMK25a}, an ultrametric $d$ on $\DNS$ was defined as follows. For $\alpha,\beta\in\NS$, let $\alpha\equiv_n\beta$ if $\alpha(x)=\beta(x)$ for all $x\in\Sigma^{<n}$, and let $[\alpha]_n$ be the $\equiv_n$-equivalence class of $\alpha$. Then for $\mu\in\SNS$, define
\begin{align*}
\norm \mu = \begin{cases}
2^n, & \text{\parbox[t]{10cm}{if $n$ is the least number such that $\mu([\alpha]_n)\ne 0$ for some $\alpha\in\NS$,}}\\
0, & \text{if no such $n$ exists.}
\end{cases}
\end{align*}
Then $d(\mu,\nu) = \norm{\mu-\nu}$.

We can extend $\norm-$ to functions $f:S\to\SNS$ by taking the $\sup$-norm $\norm f = \max_{s\in S} \norm{f(s)}$. We can say $\max$ instead of $\sup$, since all values are of the form $2^{-n}$ for some $n\in\naturals$ or $0$. The corresponding metric $d(f,g) = \norm{f-g}$ is also an ultrametric and restricts to the one defined in \cite{OMK25a} on $S\to\DNS$.

Let $(S,\struct)$ be a coalgebra with structure map $\struct:S\to\reals[\beps,\Sigma\cd S]$.
Any set function $f:S\to\SNS$ extends uniquely to a $\SR$-algebra homomorphism $f^\ddag:FS\to\SNS$. Let
\begin{align*}
& \tau:(S\to\SNS)\to(S\to\SNS) && \tau(f) = f^\ddag\circ\struct.
\end{align*}
\begin{align*}
\begin{tikzpicture}[->, >=stealth', auto]
\small
\node (o) {};
\node (NW) at (120:9mm) {$S$};
\node (E) at (0:18mm) {$\SNS$};
\node (SW) at (240:9mm) {$\reals[\beps,\Sigma\cd S]$};
\path (NW) edge node {$\tau(f)$} (E);
\path (NW) edge node[swap] {$\struct$} (SW);
\path (SW) edge node[swap] {$f^\ddag$} (E);
\end{tikzpicture}
\end{align*}

\begin{lemma}
\label{lem:contractive}
With respect to the norm $\norm-$, $\tau$ is contractive with constant of contraction 1/2.
\end{lemma}

\begin{proof}
It suffices to show that if $p\in\reals[\beps,\Sigma\cd S]$, then $\norm{f^\ddag(p)} \le \frac 12\norm f$. Then for all $s$, $\norm{\tau(f)(s)} = \norm{f^\ddag(\struct(s))} \le \frac 12\norm f$, so $\norm{\tau(f)} \le \frac 12\norm f$.

We show more generally that for $p\in\reals[\Sigma^*\cd\beps,\Sigma^{\ge m}\cd S]$, $\norm{f^\ddag(p)} \le 2^{-m}\norm f$. Equivalently, if $f(s)\equiv_n 0$ for all $s\in S$, then $f^\ddag(p) \equiv_{n+m} 0$.

We first show this for monomials in $\reals[\Sigma^{\ge m}\cd S]$. Such monomials are of the form $\prod_{x\cd s\in A}x\cd s$, where $A$ is a multiset of pairs $x\cd s$ with $\len x\ge m$. It was shown in \cite{OMK25a} that for $\alpha=\alpha\rest n$,
\begin{align}
(\mu\amp\nu)([\alpha]_n) = \sum_{\beta\uplus\gamma=\alpha}\mu([\beta]_n)\cdot\nu([\gamma]_n).
\label{eq:blah1}
\end{align}
More generally, for $A$ a multiset of signed measures,
\begin{align}
(\bigamp A)([\alpha]_{n})
&= \sum_{\biguplus_{\mu\in A}\beta_\mu=\alpha}\prod_{\mu\in A}\mu([\beta_\mu]_{n}).
\label{eq:blah2}
\end{align}
It was also shown in \cite{OMK25a} that
\begin{align}
(x\cd\mu)([\beta]_{n+\len x})
&= \begin{cases}
\mu([\alpha]_n), & \text{if $\beta=x\cd\alpha$,}\\
0, & \text{otherwise.}
\end{cases}
\label{eq:blah3}
\end{align}
Then for $\alpha=\alpha\rest n$, since $f^\ddag$ is an $\RRAlg$ morphism,
\begin{align*}
& f^\ddag(\prod_{x\cd s\in A}x\cd s)([\alpha]_{n+m})\\
&= (\bigamp_{x\cd s\in A}x\cd f(s))([\alpha]_{n+m})\\
&= \sum_{\biguplus_{x\cd s\in A}\beta_{x\cd s}=\alpha}\prod_{x\cd s\in A}(x\cd f(s))([\beta_{x\cd s}]_{n+m}) && \text{by \eqref{eq:blah2}}\\
&= \sum_{\biguplus_{x\cd s\in A}\beta_{x\cd s}=\alpha}\prod_{x\cd s\in A}\begin{cases}
f(s)([\gamma]_{n+m-\len x}), & \text{if $\beta_{x\cd s}=x\cd\gamma$,}\\
0, & \text{otherwise}
\end{cases} && \text{by \eqref{eq:blah3}}\\
&= 0,
\end{align*}
the last inference by the assumption that $f(s)\equiv_n 0$.

Finally, for $p\in\reals[\Sigma^*\cd\beps,\Sigma^{\ge m}\cd S]$, $p$ can be written as $\sum_k p_kq_k$, where $p_k\in\reals[\Sigma^*\cd\beps]$ and $q_k\in\monom[\Sigma^{\ge m}\cd S]$. By the above argument, $f^\ddag(q_k)\equiv_{n+m} 0$, and from \eqref{eq:blah1} it follows that if $\nu\equiv_{n+m} 0$, then $\mu\amp\nu\equiv_{n+m} 0$. Then
\begin{align*}
f^\ddag(p)
&= f^\ddag(\sum_k p_kq_k)
= \sum_k f^\ddag(p_k)\amp f^\ddag(q_k) \equiv_{n+m} 0.
\extagqed
\end{align*}
\end{proof}

As shown in \cite{OMK25a}, $\SNS$ is complete under the metric induced by $\norm-$, thus by the Banach fixpoint theorem, there is a unique fixpoint $\sem-$ such that $\sem- = \tau(\sem-) = \sem-^\ddag\circ\struct$.
\begin{align}
\begin{array}c
\begin{tikzpicture}[->, >=stealth', auto]
\small
\node (o) {};
\node (NW) at (120:9mm) {$S$};
\node (E) at (0:18mm) {$\SNS$};
\node (SW) at (240:9mm) {$\reals[\SeSSc]$};
\path (NW) edge node {$\sem-$} (E);
\path (NW) edge node[swap] {$\struct$} (SW);
\path (SW) edge node[swap] {$\sem-^\ddag$} (E);
\end{tikzpicture}
\end{array}
\label{eq:semantics}
\end{align}
This is the analog of the diagram (10) of \cite{OMK25a}.
Since $\sem-^\ddag = \sigma \circ F\sem-$, where $\sigma:FS\to S$ is the structure map of $\SNS$ as an E-M algebra for $F$, we can redraw this more conventionally as
\begin{align*}
\begin{tikzpicture}[->, >=stealth', auto]
\small
\node (NW) {$S$};
\node (NE) [right of=NW, node distance=22mm] {$\SNS$};
\node (SW) [below of=NW, node distance=11mm] {$FS$};
\node (SE) [below of=NE, node distance=11mm] {$F(\SNS)$};
\path (NW) edge node {$\sem-$} (NE);
\path (NW) edge node[swap] {$\struct$} (SW);
\path (SW) edge node[swap] {$F\sem-$} (SE);
\path (SE) edge node[swap] {$\sigma$} (NE);
\end{tikzpicture}
\end{align*}

From Lemma \ref{lem:prodf}, $\sem-^\ddag = (\sem-^\ddag\circ\struct)^\ddag = \sem-^\ddag\circ\struct^\ddag$, so $\sem-^\ddag$ is a fixpoint of $f\mapsto f\circ\struct^\ddag$. We henceforth dispense with the notation $^\ddag$ and write $\sem-$ for $\sem-^\ddag$ and $\struct$ for $\struct^\ddag$.

\subsection{Extension of the Structure Map}
\label{sec:extension}

Let us define an $\reals$-algebra morphism $\struct_A:\reals[S]\to\free$ for each finite tree $A\subs\Sigma^*$ by induction, maintaining the invariant $\struct_A(p)\in\reals[A\cd\beps,\del A\cd S]$ for $p\in\reals[S]$. Intuitively, the polynomial $\struct_A(s)\in\reals[A\cd\beps,\del A\cd S]$ represents the state of the computation after running the automaton starting with a single agent on state $s\in S$ and pausing after generating a string in $\del A$.

Let $\struct_\emptyset$ be the identity function. Now suppose $\struct_A$ has been defined. For $x\in\del A$, let $\struct_{A\cup\{x\}} = (x\cd\struct)\circ\struct_A$. For $p\in\reals[S]$, if $\struct_A(p)\in\reals[A\cd\beps,\del A\cd S]$ and $x\in\del A$, then
\begin{align*}
(x\cd\struct)(\struct_A(p)) &\in \reals[A\cd\beps,(\del A\setminus x)\cd S,x\cd\beps,x\Sigma\cd S]\\
&= \reals[(A\cup\{x\})\cd\beps,\del(A\cup\{x\})\cd S],
\end{align*}
so the invariant is maintained. The order by which we build $\struct_A$ does not matter by Lemma \ref{lem:cdprops}(iv).

\begin{lemma}
\label{lem:transcomp}
Let $A$ and $B$ be finite trees.
For $x\in\del A$, $\struct_{A\cup xB} = (x\cd\struct_B)\circ\struct_A$.
\end{lemma}
\begin{proof}
By induction on $\len B$. For the basis $B=\emptyset$, $\struct_{A\cup x\emptyset} = \struct_A = (x\cd\struct_\emptyset)\circ\struct_A $.
For the induction step with $y\in\del B$, using Lemma \ref{lem:cdprops}(ii),
\begin{align*}
\struct_{A\cup x(B\cup\{y\})}
&= \struct_{A\cup xB\cup\{xy\}}
= (xy\cd\struct)\circ\struct_{A\cup xB}\\
&= (xy\cd\struct)\circ(x\cd\struct_B)\circ\struct_A
= (x\cd((y\cd\struct)\circ\struct_B))\circ\struct_A\\
&= (x\cd\struct_{B\cup\{y\}})\circ\struct_A.
\extagqed
\end{align*}
\end{proof}

\begin{lemma}
\label{lem:iteratestruct}
For all finite trees $A$, $\sem- = \sem-\circ\struct_A$.
\end{lemma}

\begin{proof}
By induction. It suffices to show the two sides agree on inputs $s\in S$. We have $\sem- = \sem-\circ\struct_\emptyset$, since $\struct_\emptyset$ is the identity function, and $\sem- = \sem-\circ\struct = \sem-\circ\struct_{\{\eps\}}$ as a consequence of Lemma \ref{lem:contractive}.

Now suppose the lemma holds for finite trees $A$ and $B$. For $x\in\del A$, on inputs $s\in S$,
\begin{align*}
& \sem-\circ\struct_{A\cup xB}\\
&= \sem-\circ(x\cd\struct_B)\circ\struct_A && \text{Lemma \ref{lem:transcomp}}\\
&= (\del A\cd\sem-)\circ(x\cd\struct_B)\circ\struct_A && \text{Lemma \ref{lem:cdprops}(iii)}\\
&= ((\del A\setminus x)\cd\sem-)\circ(x\cd\sem-)\circ(x\cd\struct_B)\circ\struct_A && \text{Lemma \ref{lem:cdprops}(iv)}\\
&= ((\del A\setminus x)\cd\sem-)\circ(x\cd(\sem-\circ\struct_B))\circ\struct_A && \text{Lemma \ref{lem:cdprops}(ii)}\\
&= ((\del A\setminus x)\cd\sem-)\circ(x\cd\sem-)\circ\struct_A && \text{induction hypothesis on $B$}\\
&= (\del A\cd\sem-)\circ\struct_A && \text{Lemma \ref{lem:cdprops}(iv) again}\\
&= \sem-\circ\struct_A && \text{Lemma \ref{lem:cdprops}(iii) again}\\
&= \sem- && \text{induction hypothesis on $A$}.
\extagqed
\end{align*}
\end{proof}

\subsection{Marginalization}
\label{sec:marginalization}

\emph{Marginalization} refers to the projection of a joint distribution on a product space onto its marginal distribution on a factor space. This is achieved by partial evaluation of the measure on the complementary factor. For automata, we can intuitively marginalize on a finite tree $A$ by running the automaton until reaching all strings in $\del A$ and then pruning the rest of the computation tree by outright failing. Intuitively, after an agent has scanned a string in $\del A$, no further computation of that agent can contribute to the acceptance of any string in $A$. In this section we show how to make this idea precise in the context of $\SR$-algebras.

For $\mu\in\SNS$ and $B$ a measurable subset of $\NS$, define
\begin{align*}
\eval_B(\mu) = \mu(B).
\end{align*}

\begin{lemma}
\label{lem:eval}
The map $\eval_{\NS}:\SNS\to\reals$ is a $\SRAlg$ morphism, regarding $\reals$ as a $\SR$-algebra with $x\cd r = r$ and $\beps=1$.
\end{lemma}
\begin{proof}
Every $\eval_B$ is linear, as $(a\mu+b\nu)(B)=a\mu(B)+b\nu(B)$ by definition. Not every $\eval_B$ is a homomorphism with respect to multiplication, but $\eval_{\NS}$ is:
\begin{align*}
(\mu\amp\nu)(\NS)
&= (\mu\otimes\nu)(\uplus^{-1}(\NS))
= (\mu\otimes\nu)(\NS\times\NS)
= \mu(\NS)\cdot\nu(\NS).
\end{align*}
Finally, $\eval_{\NS}$ is a homomorphism with respect to the monoid action:
\begin{align*}
(x\cd\mu)(\NS)
&= \mu((x\cd-)^{-1}(\NS))
= \mu(\NS)
= x\cd\mu(\NS).
\extagqed
\end{align*}
\end{proof}

Let $\bone:\free\to\reals$ be the unique $\SR$-algebra homomorphism such that $\bone(s) = 1$ for all $s\in S$. Note that $\sem s(\NS) = 1$ for all $s\in S$, as $\sem s$ is a probability distribution, thus $\eval_{\NS}\circ\sem-$ and $\bone$ agree on $S$. Since the extension to a $\SR$-algebra morphism on domain $\free$ is unique,
\begin{align}
& \eval_{\NS}\circ\sem- = \bone : \free\to\reals.\label{eq:evalisone}
\end{align}

Let $A$ be a finite tree. For $p\in\reals[S]$, rewriting $\struct_A(p)$ in the form $\reals[\del A\cd S][A\cd\beps]$ defines a collection of polynomials $d_A^m(p)$ such that
\begin{align}
\struct_A(p)
&= \sum_{m\in\monom[A\cd\beps]} d_A^m(p)\cdot m.
\label{eq:ddef}
\end{align}
The function $d_A^m:\reals[S]\to\reals[\del A\cd S]$ for $m\in\monom[A\cd\beps]$ is uniquely defined in this way.
Since $\struct_A(p)$ is a polynomial, all but finitely many of the $d_A^m(p)$ are $0$.

Let us further define for $x\in\del A$ the unique function $d_{A,x}^m:\reals[S]\to\reals[S]$ such that
\begin{align}
((\del A\setminus x)\cd\bone)\circ d_A^m)(p) &= x\cd d_{A,x}^m(p).
\label{eq:dmdef}
\end{align}
The left-hand side is of type $\reals[x\cd S]$, and since $x\cd-$ is injective, this defines $d_{A,x}^m(p)$ uniquely. 

We remark that, although the maps $\struct_A:\reals[S]\to\reals[A\cd\beps,\del A\cd S]$ are $\reals$-algebra homomorphisms, the maps $d_{A}^m:\reals[S]\to\reals[\del A\cd S]$ and $d_{A,x}^m:\reals[S]\to\reals[S]$ are not in general. However, they do satisfy some useful algebraic properties.
\begin{lemma}
\label{lem:dna}
Let $A$ be a finite tree, $m\in\monom[A\cd\beps]$, $r\in\reals$, and $p,q\in\reals[S]$.
The following equational properties hold of $d_{A}^m$.
\begin{align*}
d_{A}^m(ap+bq) &= ad_{A}^m(p) + bd_{A}^m(q)\\
d_{A}^m(pq) &= \sum_{k\ell=m} d_{A}^k(p)\cdot d_{A}^\ell(q)\\
d_{A}^1(r) &= r\\
d_{A}^m(r) &= 0,\ m\ne 1.
\end{align*}
Moreover, the analogous properties hold for the functions $d_{A,x}^m$.
\end{lemma}

\begin{proof}
For linearity, we have
\begin{align*}
\struct_A(ap+bq) &= \sum_{m\in\monom[A\cd\beps]} d_A^m(ap+bq)\cdot m,
\end{align*}
but since $\struct_A$ is an $\reals$-algebra homomorphism, we also have
\begin{align*}
& \struct_A(ap+bq) = a\struct_A(p) + b\struct_A(q)\\
&= a\cdot\sum_{m\in\monom[A\cd\beps]} d_{A}^m(p)\cdot m + b\cdot\sum_{m\in\monom[A\cd\beps]} d_{A}^m(q)\cdot m\\
&= \sum_{m\in\monom[A\cd\beps]} (ad_{A}^m(p) + bd_{A}^m(q))\cdot m,
\end{align*}
therefore $d_{A}^m(ap+bq) = ad_{A}^m(p) + bd_{A}^m(q)$.

For products, again we have
\begin{align*}
\struct_A(pq) = \sum_{m\in\monom[A\cd\beps]} d_{A}^m(pq)\cdot m
\end{align*}
and
\begin{align*}
\struct_A(pq) &= \struct_A(p)\cdot \struct_A(q)
= (\sum_{k\in\monom[A\cd\beps]} d_{A}^k(p)\cdot k)\cdot(\sum_{\ell\in\monom[A\cd\beps]} d_{A}^\ell(q)\cdot\ell)\\
&= \sum_{k\in\monom[A\cd\beps]}\sum_{\ell\in\monom[A\cd\beps]} d_{A}^k(p)\cdot d_{A}^\ell(q)\cdot k\ell\\
&= \sum_{m\in\monom[A\cd\beps]}\sum_{k\ell=m} d_{A}^k(p)\cdot d_{A}^\ell(q)\cdot m\\
&= \sum_{m\in\monom[A\cd\beps]}(\sum_{k\ell=m} d_{A}^k(p)\cdot d_{A}^\ell(q))\cdot m,
\end{align*}
therefore $d_{A}^m(pq) = \sum_{k\ell=m} d_{A}^k(p)\cdot d_{A}^\ell(q)$.

For scalars $r\in\reals$, we have $r = d_{A}(r) = \sum_{m\in\monom[A\cd\beps]} d_{A}^m(r)\cdot m$, so $d_{A}^m(r)=0$ for $m\ne 1$ and $d_{A}^1(r)=r$.

The same equational properties are easily seen to hold also for $d_{A,X}^m$, since $(\del A\setminus x)\cd\bone$ and $x\cd-$ behave homomorphically and $x\cd-$ is injective. For example, the argument for products is
\begin{align*}
x\cd d_{A,x}^m(pq)
&= ((\del A\setminus x)\cd\bone)(d_{A}^m(pq))
= ((\del A\setminus x)\cd\bone)(\sum_{k\ell=m} d_{A}^k(p)\cdot d_A^\ell(q))\\
&= \sum_{k\ell=m} ((\del A\setminus x)\cd\bone)(d_{A}^k(p))\cdot((\del A\setminus x)\cd\bone)(d_A^\ell(q))\\
&= \sum_{k\ell=m} (x\cd d_{A,x}^k(p))\cdot(x\cd d_{A,x}^\ell(q))
= x\cd\sum_{k\ell=m} d_{A,x}^k(p)\cdot d_{A,x}^\ell(q),
\end{align*}
therefore $d_{A,x}^m(pq) = \sum_{k\ell=m} d_{A,x}^k(p)\cdot d_{A,x}^\ell(q)$ by injectivity.
\exqed
\end{proof}

\begin{lemma}
\label{lem:cc}
For $A,B$ finite trees, $x\in\del A$, $y\in\del B$, $m\in\monom[A\cd\beps]$, and $\ell\in\monom[B\cd\beps]$,
$d_{A\cup xB,xy}^{(x\cd\ell)\cdot m} = d_{B,y}^\ell\circ d_{A,x}^m$.
\end{lemma}

\begin{proof}
Since $\del(A\cup xB) = (\del A\setminus x)\cup x\del B$ and the elements are prefix-incomparable, by Lemma \ref{lem:cdprops}(iv) and (i) we have
\begin{align*}
(\del(A\cup xB)\setminus xy)\cd\bone
&= ((x\del B\setminus xy)\cd\bone)\circ((\del A\setminus x)\cd\bone)\\
&= (x\cd((\del B\setminus y)\cd\bone)\circ((\del A\setminus x)\cd\bone).
\end{align*}
By Lemmas \ref{lem:transcomp} and \ref{lem:cdprops}(iv),
\begin{align}
& ((\del(A\cup xB)\setminus xy)\cd\bone)\circ\struct_{A\cup xB}\nonumber\\
&= (x\cd((\del B\setminus y)\cd\bone)\circ((\del A\setminus x)\cd\bone)\circ(x\cd\struct_B)\circ\struct_A\nonumber\\
&= (x\cd((\del B\setminus y)\cd\bone)\circ(x\cd\struct_B)\circ((\del A\setminus x)\cd\bone)\circ\struct_A.\label{eq:plain}
\end{align}
Now let us apply both sides to an arbitrary $p\in\reals[S]$. By \eqref{eq:ddef} and \eqref{eq:dmdef}, the left-hand side gives
\begin{align*}
& ((\del(A\cup xB)\setminus xy)\cd\bone)(\struct_{A\cup xB}(p))\\
&= ((\del(A\cup xB)\setminus xy)\cd\bone)(\sum_{k\in\monom[(A\cup xB)\cd\beps]} d_{A\cup xB}^k(p)\cdot k)\\
&= \sum_{k\in\monom[(A\cup xB)\cd\beps]} ((\del(A\cup xB)\setminus xy)\cd\bone)(d_{A\cup xB}^k(p))\cdot((\del(A\cup xB)\setminus xy)\cd\bone)(k)\\
&= \sum_{k\in\monom[(A\cup xB)\cd\beps]} (xy\cd d_{A\cup xB,xy}^k(p))\cdot k.
\end{align*}
Since every monomial $k\in\monom[(A\cup xB)\cd\beps]$ can be uniquely factored as $k = (x\cd\ell)\cdot m$ for $m\in\monom[A\cd\beps]$ and $\ell\in\monom[B\cd\beps]$, this is equal to
\begin{align}
\sum_{m\in\monom[A\cd\beps]} \sum_{\ell\in\monom[B\cd\beps]} (xy\cd d_{A\cup xB,xy}^{(x\cd\ell)\cdot m}(p))\cdot(x\cd\ell)\cdot m.\label{eq:plainA}
\end{align}

For the right-hand side of \eqref{eq:plain}, we apply the functions from right to left. By \eqref{eq:ddef} and \eqref{eq:dmdef}, 
\begin{align*}
((\del A\setminus x)\cd\bone)(\struct_A(p))
&= ((\del A\setminus x)\cd\bone)(\sum_{m\in\monom[A\cd\beps]} d_A^m(p)\cdot m)\\
&= \sum_{m\in\monom[A\cd\beps]} ((\del A\setminus x)\cd\bone)(d_A^m(p))\cdot ((\del A\setminus x)\cd\bone)(m)\\
&= \sum_{m\in\monom[A\cd\beps]} (x\cd d_{A,x}^m(p))\cdot m.
\end{align*}
Now applying $x\cd\struct_B$ to this, by \eqref{eq:ddef},
\begin{align*}
& (x\cd\struct_B)(\sum_{m\in\monom[A\cd\beps]} (x\cd d_{A,x}^m(p))\cdot m)\\
&= \sum_{m\in\monom[A\cd\beps]} (x\cd\struct_B)(x\cd d_{A,x}^m(p))\cdot (x\cd\struct_B)(m)\\
&= \sum_{m\in\monom[A\cd\beps]} (x\cd \struct_B(d_{A,x}^m(p)))\cdot m\\
&= \sum_{m\in\monom[A\cd\beps]} (x\cd \sum_{\ell\in\monom[B\cd\beps]} d_B^\ell(d_{A,x}^m(p))\cdot\ell)\cdot m\\
&= \sum_{m\in\monom[A\cd\beps]} \sum_{\ell\in\monom[B\cd\beps]} (x\cd d_B^\ell(d_{A,x}^m(p)))\cdot(x\cd\ell)\cdot m.
\end{align*}
Finally, applying $x\cd((\del B\setminus y)\cd\bone)$, by \eqref{eq:dmdef},
\begin{align}
& (x\cd((\del B\setminus y)\cd\bone))(\sum_{m\in\monom[A\cd\beps]} \sum_{\ell\in\monom[B\cd\beps]} (x\cd d_B^\ell(d_{A,x}^m(p)))\cdot(x\cd\ell)\cdot m)\nonumber\\
&= \sum_{m\in\monom[A\cd\beps]} \sum_{\ell\in\monom[B\cd\beps]} (x\cd((\del B\setminus y)\cd\bone))(x\cd d_B^\ell(d_{A,x}^m(p)))\cdot(x\cd((\del B\setminus y)\cd\bone))((x\cd\ell)\cdot m)\nonumber\\
&= \sum_{m\in\monom[A\cd\beps]} \sum_{\ell\in\monom[B\cd\beps]} (x\cd ((\del B\setminus y)\cd\bone)(d_B^\ell(d_{A,x}^m(p))))\cdot(x\cd\ell)\cdot m\nonumber\\
&= \sum_{m\in\monom[A\cd\beps]} \sum_{\ell\in\monom[B\cd\beps]} (x\cd (y\cd d_{B,y}^\ell(d_{A,x}^m(p))))\cdot(x\cd\ell)\cdot m\nonumber\\
&= \sum_{m\in\monom[A\cd\beps]} \sum_{\ell\in\monom[B\cd\beps]} (xy\cd d_{B,y}^\ell(d_{A,x}^m(p)))\cdot(x\cd\ell)\cdot m.
\label{eq:plainB}
\end{align}

Now the equality \eqref{eq:plain} implies the equality of \eqref{eq:plainA} and \eqref{eq:plainB}, thus
\begin{align*}
xy\cd d_{A\cup xB,xy}^{(x\cd\ell)\cdot m}(p)
&= xy\cd d_{B,y}^\ell(d_{A,x}^m(p)).
\end{align*}
Since the monoid action is injective and $p$ was arbitrary, the result follows.
\exqed
\end{proof}

\subsection{Behavioral Equivalence}
\label{sec:behavequivalence}

The following lemma and theorem establish necessary and sufficient conditions for behavioral equivalence.

Recalling the bijective correspondence between monomials $\monom[A\cd\beps]$ and multisets $\naturals^A$, the monomial corresponding to the multiset $\alpha\in\naturals^A$ is
\begin{align}
m(\alpha) &= \prod_{y\in A} (y\cd\beps)^{\alpha(y)}.\label{eq:monommultiset}
\end{align}

\begin{lemma}
\label{lem:Adetermined}
Let $A$ be a finite tree. Let $\beta\in\naturals^A$ and let $m(\beta)\in\monom[A\cd\beps]$ be the corresponding monomial. Let $d_A^m(s)$ be defined as in \eqref{eq:ddef}. Then $\bone(d_A^{m(\beta)}(s)) = \sem s([\beta]_A)$.
\end{lemma}
\begin{proof}
Using the properties of $\amp$ and $y\cd-$ in $\SNS$, we have
\begin{align*}
\sem{m(\alpha)} = \sem{\prod_{y\in A} (y\cd\beps)^{\alpha(y)}}
&= \bigamp_{y\in A} \sem{(y\cd\beps)^{\alpha(y)}}
= \bigamp_{y\in A} \bigamp_{i=1}^{\alpha(y)} y\cd\dirac{\{\eps\}}
= \bigamp_{y\in A} \bigamp_{i=1}^{\alpha(y)} \dirac{\{y\}}
= \dirac\alpha.
\end{align*}
By Lemma \ref{lem:iteratestruct},
\begin{align*}
\sem s &= \sem{\struct_A(s)}
= \sem{\sum_{\alpha\in\naturals^A} d_{A}^{m(\alpha)}(s)\cdot m(\alpha)}
= \sum_{\alpha\in\naturals^A} \sem{d_{A}^{m(\alpha)}(s)}\amp\dirac\alpha.
\end{align*}
Applying this to $[\beta]_A$ and using \eqref{eq:blah1},
\begin{align}
\sem s([\beta]_A)
&= \sum_{\alpha\in\naturals^A} (\sem{d_{A}^{m(\alpha)}(s)}\amp\dirac\alpha)([\beta]_A)\nonumber\\
&= \sum_{\alpha\in\naturals^A} \sum_{\gamma\uplus\zeta=\beta}\sem{d_{A}^{m(\alpha)}(s)}([\gamma]_A)\cdot\dirac\alpha([\zeta]_A).\label{eq:Adetermined1}
\end{align}
All indeterminates of $d_{A}^{m(\alpha)}(s)$ are of the form $x\cd s$ for some $x\in\del A$, which says that all indeterminates of $F\sem-(d_{A}^{m(\alpha)}(s))$ are of the form $x\cd\sem s$ for some $x = \del A$. This implies that the support of $\sem{d_{A}^{m(\alpha)}(s)}$ is contained in $[0]_A$, so $\sem{d_{A}^{m(\alpha)}(s)}([\gamma]_A)$ vanishes unless $\gamma=0$, in which case $\sem{d_{A}^{m(\alpha)}(s)}([0]_n)=\sem{d_{A}^{m(\alpha)}(s)}(\NS)$. Moreover, $\dirac\alpha([\zeta]_A)$ vanishes unless $\alpha=\zeta$, in which case it equals $1$. From these facts it follows that \eqref{eq:Adetermined1} reduces to $\sem{d_{A}^{m(\beta)}(s)}(\NS)$. By \eqref{eq:evalisone}, this is $\bone(\sem{d_{A}^{m(\beta)}}(s))$.
\exqed
\end{proof}

\begin{theorem}
\label{thm:necandsuf}
For an automaton with state set $S$ and transition function $\struct:S\to\reals[\beps,\Sigma\cd S]$, for $\mu,\nu\in\SNs$, let $p=\phi_S^{-1}(\mu-\nu)\in\reals[S]$. Let $d_{A}^m$ and $d_{A,x}^m$ be defined as in \eqref{eq:ddef} and \eqref{eq:dmdef}. Then $\sem\mu=\sem\nu$ if and only if for all finite trees $A\subs\Sigma^*$, all $x\in\del A$, and all $m\in\monom[A\cd\beps]$,
\begin{align}
d_{A,x}^m(p) \in \ker\bone.
\label{eq:necandsuf}
\end{align}
\end{theorem}
\begin{proof}
Let $p, A, x, m$ be as specified. We have
\begin{align*}
	\bone(d_{A}^m(p)) &= (\del A\cd\bone)(d_{A}^m(p)) && \text{Lemma \ref{lem:cdprops}(iii), definition \eqref{eq:ddef}} \\
	&= (x\cd\bone)((\del A\setminus x)\cd\bone)(d_{A}^m(p)) && \text{Lemma \ref{lem:cdprops}(iv)} \\
	&= (x\cd\bone)(x\cd d_{A,x}^m(p)) && \text{definition \eqref{eq:dmdef}} \\
	&= x\cd\bone(d_{A,x}^m(p)) && \text{definition of $x\cd -$} \\
	&= \bone(d_{A,x}^m(p)) && \text{definition of $\bone$,}
\end{align*}
thus $d_{A,x}^m(p)\in\ker\bone$ iff $d_{A}^m(p)\in\ker\bone$. By Lemma \ref{lem:Adetermined}, the latter condition is equivalent to $\sem\mu([\alpha]_A)=\sem\nu([\alpha]_A)$, where $\alpha\in\naturals^A$ is the multiset corresponding to the monomial $m$ according to the bijection \eqref{eq:monommultiset}.
But since the $[\alpha]_{A}$ generate the Borel sets, this occurs for all $A$ if and only if $\sem\mu=\sem\nu$.
\exqed
\end{proof}

\section{Algorithm}
\label{sec:algorithm}

To test whether $\sem\mu=\sem\nu$ for given $\mu,\nu\in\SNs$, we use Theorem \ref{thm:necandsuf} to check whether all polynomials $d_{A,x}^m(p)\in\reals[S]$ are members of the ideal $\ker\bone$, where $p=\phi_S^{-1}(\mu-\nu)$. Although there are infinitely many such polynomials, this can be done in finite time.

\begin{lemma}
\label{lem:cc2}
Let $I$ be an ideal of $\reals[S]$. Suppose for all $a\in\Sigma$ and $n\ge 0$, $\dean$ preserves membership in $I$; that is, if $p\in I$, then $\dean(p)\in I$. Then for all $A$, all $x\in\del A$, and all monomials $m\in\monom[A\cd\beps]$, $d_{A,x}^m$ preserves membership in $I$.
\end{lemma}
\begin{proof}
By induction. All $\dean$ preserve $I$ by assumption, and by Lemma \ref{lem:cc}, every $d_{A,x}^m$ is a composition of these; specifically, $d_{A\cup\{x\},xa}^m = d_{\eps,a}^{\beps^k}\circ d_{A,x}^\ell$, where $k$ is the degree of $x$ in $m$ and $\ell=m/x^k$. Any composition of functions that preserve $I$ also preserves $I$.
\exqed
\end{proof}

\begin{lemma}
\label{lem:preservebasis}
Suppose $I = \angle{B}$ is an ideal of $\reals[S]$ generated by $B$. Let $a\in\Sigma$.
If $\dean(b)\in I$ for all $n\ge 0$ and $b\in B$, then for all $n\ge 0$, $\dean$ preserves membership in $I$.
\end{lemma}
\begin{proof}
If $p\in I$, then $p = p_1b_1 +\cdots+ p_kb_k$ for some $\seq p1k\in\reals[S]$ and $\seq b1k\in B$. By Lemma \ref{lem:dna},
\begin{align*}
\dean(p_1b_1 +\cdots+ p_nb_n)
&= \dean(p_1b_1) +\cdots+ \dean(p_nb_n)\\
&= \sum_{m+\ell=n}d_{\eps,a}^{\beps^m}(p_1)d_{\eps,a}^{\beps^\ell}(b_1)\\
&\quad +\ \cdots\\
&\quad +\ \sum_{m+\ell=n}d_{\eps,a}^{\beps^m}(p_n)d_{\eps,a}^{\beps^\ell}(b_n) \in I.
\extagqed
\end{align*}
\end{proof}

The \emph{forward ideal} of the automaton $(S,\theta)$ with seed $p\in\reals[S]$ is
\begin{align*}
J_p &= \angle{\set{d_{A,x}^m(p)}{A\subs\Sigma^* \text{ a finite tree},\ x\in\del A,\ m\in\monom[A\cd\beps]}}\\
&\subs \reals[S].
\end{align*}
By Theorem \ref{thm:necandsuf}, the condition \eqref{eq:necandsuf} for behavioral equivalence of $\mu$ and $\nu$ is $J_{\phi^{-1}(\mu-\nu)}\subs\ker\bone$.
By Lemma \ref{lem:preservebasis}, all $\dean$ preserve membership in $J_p$, since $d_{A\cup\{x\},xa}^{mx^n}(p) = \dean(d_{A,x}^m(p))$ by Lemma \ref{lem:cc}.

\begin{lemma}
\label{lem:halting}
Suppose $I = \angle B$ is an ideal of $\reals[S]$ containing $p$ with basis $B$. If $p\in I$ and $\dean(b)\in I$ for all $a\in\Sigma$, $n\ge 0$, and $b\in B$, then $J_p\subs I$.
\end{lemma}
\begin{proof}
By Lemma \ref{lem:preservebasis}, all $\dean$ preserve membership in $I$.
By Lemma \ref{lem:cc2}, for all finite trees $A\subs\Sigma^*$, all $x\in\del A$, and all monomials $m\in\monom[A\cd\beps]$, $d_{A,x}^m$ preserves membership in $I$. Since $p\in I$, all $d_{A,x}^m(p)\in I$, thus $J_p\subs I$.
\exqed
\end{proof}

The algorithm proceeds in stages, constructing a sequence of subsets $B_0\subs B_1\subs B_2\subs\cdots$ of $\reals[S]$, maintaining the following invariants:
\begin{enumerate}
\item
$p\in\angle{B_n}$.
\item
$B_n$ is a Gr\"obner basis for $\angle{B_n}$ \cite{Buchberger76,Lazard83}. This allows for efficient membership testing in $\angle{B_n}$ \cite{MayrMeyer82}.
\item
$B_n\subs J_p$. This implies that $\angle{B_n}\subs J_p$.
\item
$B_n\subs \ker\bone$. This implies that $\angle{B_n}\subs\ker\bone$.
\end{enumerate}
We start at stage $0$ with the seed $p=\phi_S^{-1}(\mu-\nu)$ and $B_0=\{p\}$. Invariants 1 and 2 are clearly satisfied, and 3 is satisfied since $p = d_{\emptyset,\eps}^1(p)$. Invariant 4 may fail, but we can test it by computing $\bone(p)$, that is, by evaluating $p$ at the point $(1,\ldots,1)$. If the result is nonzero, we halt and declare $\mu$ and $\nu$ not equivalent, as $p$ is a witness for $J_p\not\subs\ker\bone$. Intuitively, this will happen if for some $k\ge 0$, the probability with respect to $\mu$ that $\eps$ is accepted with multiplicity $k$ differs from the probability of the same event with respect to $\nu$. If $\bone(p)=0$, then invariant 4 is satisfied.

Now suppose at stage $n$, we have constructed $B_n$ satisfying the four invariants. For each $b\in B_n$, we compute $\dean(b)$ for all $a\in\Sigma$ and $n\ge 0$. This can be done all at once by computing $\struct(b)$, then applying $(\Sigma\setminus a)\cd\bone$ for each $a\in\Sigma$ to obtain the polynomial $\sum_n (a\cd d^{\beps^n}_{\eps,a}(b))\cdot\beps^n$, from which $\dean(b)$ can be read off. All but finitely many of these are $0$.

For all the nonzero $\dean(b)$, we test whether $\dean(b)\in\ker\bone$ by evaluating $\bone(\dean(b))$. Again, we can do this all at once simply by evaluating $\struct(b)$ at $(1,\ldots,1)$. If the result is nonzero, we halt immediately and report $\mu$ and $\nu$ not equivalent. This happens iff some $\bone(\dean(b))\ne 0$, which by invariant 3 and the fact that membership in $J_p$ is preserved by $\dean$ says that $J_p\not\subs\ker\bone$.

If this test succeeds, then we have a new finite set $\dean(b)\in\ker\bone$ for $n\ge 0$, $a\in\Sigma$, and $b\in B_n$ to add to our collection. Two things can happen: either all the new $\dean(b)$ are already contained in $\angle{B_n}$, or at least one is not. We test this for each $\dean(b)$ using Buchberger's algorithm \cite{Buchberger76,Lazard83,MayrMeyer82}.

If all $\dean(b)\in\angle{B_n}$, we are done. In this case, we halt and declare $\mu$ and $\nu$ behaviorally equivalent. This holds because by Lemma \ref{lem:halting} and invariants 1, 3, and 4, $\angle{B_n} = J_p\subs\ker\bone$.

If not all $\dean(b)\in\angle{B_n}$, then we add the new ones to $B_n$, recompute a Gr\"obner basis $B_{n+1}$ for this new set of polynomials, and go on to the next stage. The invariants are maintained.

The algorithm must halt after finitely many stages with either success or failure due to the fact that the multivariate polynomial ring $\reals[S]$ is Noetherian, which means there is no infinite proper ascending chain of ideals.

We have shown
\begin{theorem}
It is decidable whether two automata or two expressions are semantically equivalent.
\end{theorem}

\section{Conclusion and Open Problems}
\label{sec:conclusion}

We have exploited a connection between polynomial algebra and state-based systems with probability and  nondeterminism represented by multisets to give a decision procedure for semantic equivalence of automata or expressions. The algorithm uses Gr\"obner basis techniques, but is otherwise reminiscent of the decision procedure for weighted automata.

We have made no attempt to analyze the complexity of our algorithm, as it is probably not worthwhile to do so given the high complexity of Gr\"obner basis computations. The worst-case complexity is likely quite bad, at least exponential space \cite{MayrMeyer82}. However there may be special cases in which the complexity is much lower, and this is a subject of current investigation. Most promisingly, a guarded version of the language in which nondeterminism is constrained to appear in the form of conditional tests $\ite bpq = bp \amp \bar bq$ and while loops $\whd bp = \fix x{\ite b{p\cmp x}\skp}$ is likely to be much more tractable, as the constraints mitigate the proliferation of state due to nondeterminism. This would require the introduction of Boolean tests as in \KAT\ \cite{K97c} and a language similar to guarded \KAT\ (\GKAT) \cite{SFHKKS20a}. The equivalence problem for \GKAT\ is nearly linear time, whereas the same problem for general \KAT\ is \PSPACE-complete. An analogous speedup is likely to be realized here.

Another major open question is completeness. In \cite{OMK25a} a set of sound equational axioms were presented for reasoning about program equivalence, but the completeness question was left open. The axioms are given in \S\ref{apx:equations}. That paper did not have the advantage of the connection to polynomial algebra, and armed with this new insight, we plan to explore this question vigorously.

\nocite{KS24a}


\bibliographystyle{splncs04}
\bibliography{dk,Decision,PKAT}

@inproceedings{Lazard83,
author = {Lazard, Daniel},
title = {Gr\"{o}bner bases, {G}aussian elimination and resolution of systems of algebraic equations},
year = {1983},
isbn = {978-3-540-12868-7},
doi = {10.1007/3-540-12868-9-99},
publisher = {Springer-Verlag},
address = {Berlin, Heidelberg},
booktitle = {Proc. European Conference on Computer Algebra (EUROCAL '83)},
pages = {146--156},
series = {Lecture notes in Computer Science},
volume = {162}
}

@article{MayrMeyer82,
title = {The complexity of the word problems for commutative semigroups and polynomial ideals},
journal = {Advances in Mathematics},
volume = {46},
number = {3},
pages = {305-329},
year = {1982},
issn = {0001-8708},
doi = {https://doi.org/10.1016/0001-8708(82)90048-2},
url = {https://www.sciencedirect.com/science/article/pii/0001870882900482},
author = {Ernst W Mayr and Albert R Meyer},
}

@article{Buchberger76,
author = {Buchberger, B.},
title = {A theoretical basis for the reduction of polynomials to canonical forms},
year = {1976},
issue_date = {August 1976},
publisher = {Association for Computing Machinery},
address = {New York, NY, USA},
volume = {10},
number = {3},
issn = {0163-5824},
url = {https://doi.org/10.1145/1088216.1088219},
doi = {10.1145/1088216.1088219},
journal = {SIGSAM Bull.},
month = aug,
pages = {19–29},
numpages = {11}
}

@incollection{Mohri09,
author = {Mehryar Mohri},
title = {Weighted Automata Algorithms},
editor = {Droste, Manfred and Kuich, Werner and Vogler, Heiko},
pages ={213--254},
booktitle = {Handbook of Weighted Automata},
series = {Monographs in Theoretical Computer Science},
year = {2009},
isbn = {978-3-642-01491-8},
publisher = {Springer},
edition = {1st}
}

@article{Kiefer20,
  author       = {Stefan Kiefer},
  title        = {Notes on Equivalence and Minimization of Weighted Automata},
  journal      = {CoRR},
  volume       = {abs/2009.01217},
  year         = {2020},
  url          = {https://arxiv.org/abs/2009.01217},
  eprinttype    = {arXiv},
  eprint       = {2009.01217},
  bibsource    = {dblp computer science bibliography, https://dblp.org}
}

@article{DahlqvistNeves18,
  author       = {Fredrik Dahlqvist and
                  Renato Neves},
  title        = {Compositional semantics for new paradigms: probabilistic, hybrid and
                  beyond},
  journal      = {CoRR},
  volume       = {abs/1804.04145},
  year         = {2018},
  url          = {http://arxiv.org/abs/1804.04145},
  eprinttype    = {arXiv},
  eprint       = {1804.04145},
  bibsource    = {dblp computer science bibliography, https://dblp.org}
}

@inproceedings{BeckettCampbellGreenberg17,
author = {Greenberg, Michael and Beckett, Ryan and Campbell, Eric},
title = {Kleene algebra modulo theories: a framework for concrete KATs},
year = {2022},
isbn = {9781450392655},
publisher = {Association for Computing Machinery},
address = {New York, NY, USA},
url = {https://doi.org/10.1145/3519939.3523722},
doi = {10.1145/3519939.3523722},
booktitle = {Proceedings of the 43rd ACM SIGPLAN International Conference on Programming Language Design and Implementation},
pages = {594–608},
location = {San Diego, CA, USA},
series = {PLDI 2022}
}

@InProceedings{Kiefer11,
author="Kiefer, Stefan
and Murawski, Andrzej S.
and Ouaknine, Jo{\"e}l
and Wachter, Bj{\"o}rn
and Worrell, James",
editor="Gopalakrishnan, Ganesh
and Qadeer, Shaz",
title="Language Equivalence for Probabilistic Automata",
booktitle="Computer Aided Verification",
year="2011",
publisher="Springer Berlin Heidelberg",
address="Berlin, Heidelberg",
pages="526--540",
isbn="978-3-642-22110-1"
}

@article{Tzeng92,
author = {Tzeng, Wen-Guey},
title = {A Polynomial-Time Algorithm for the Equivalence of Probabilistic Automata},
journal = {SIAM Journal on Computing},
volume = {21},
number = {2},
pages = {216-227},
year = {1992},
doi = {10.1137/0221017},
URL = {https://doi.org/10.1137/0221017},
eprint = {https://doi.org/10.1137/0221017}
}

@inproceedings{Jacobs21,
  author={Jacobs, Bart},
  booktitle={2021 36th Annual ACM/IEEE Symp. Logic in Computer Science (LICS)}, 
  title={From Multisets over Distributions to Distributions over Multisets}, 
  year={2021},
  volume={},
  number={},
  pages={1-13},
  doi={10.1109/LICS52264.2021.9470678}}

@phdthesis{Zwart20,
author = "Maaike Zwart",
title = "On the non-compositionality of monads via distributive laws",
school = "Oxford University",
year = "2020"
}

@phdthesis{Varacca03,
author = "Daniele Varacca",
title = "Probability, nondeterminism and concurrency: Two denotational models for probabilistic computation",
school = "Aarhus University",
year = "2003"
}

@inproceedings{GoyPetrisan20,
  author    = {Alexandre Goy and Daniela Petrisan},
  editor    = {Holger Hermanns and
               Lijun Zhang and
               Naoki Kobayashi and
               Dale Miller},
  title     = {Combining probabilistic and non-deterministic choice via weak distributive
               laws},
  booktitle = {{LICS} '20: 35th Annual {ACM/IEEE} Symposium on Logic in Computer
               Science, Saarbr{\"{u}}cken, Germany, July 8-11, 2020},
  pages     = {454--464},
  publisher = {{ACM}},
  year      = {2020},
  url       = {https://doi.org/10.1145/3373718.3394795},
  doi       = {10.1145/3373718.3394795},
  bibsource = {dblp computer science bibliography, https://dblp.org}
}

@inproceedings{DahlqvistParlantSilva18,
  author    = {Fredrik Dahlqvist and Louis Parlant and Alexandra Silva},
  editor    = {Bernd Fischer and
               Tarmo Uustalu},
  title     = {Layer by Layer - Combining Monads},
  booktitle = {Theoretical Aspects of Computing - {ICTAC} 2018 - 15th International
               Colloquium, Stellenbosch, South Africa, October 16-19, 2018, Proceedings},
  series    = {Lecture Notes in Computer Science},
  volume    = {11187},
  pages     = {153--172},
  publisher = {Springer},
  year      = {2018},
  url       = {https://doi.org/10.1007/978-3-030-02508-3\_9},
  doi       = {10.1007/978-3-030-02508-3\_9},
  bibsource = {dblp computer science bibliography, https://dblp.org}
}

@article{ZwartMarsden22,
  author    = {Maaike Zwart and Dan Marsden},
  title     = {No-Go Theorems for Distributive Laws},
  journal   = {Log. Methods Comput. Sci.},
  volume    = {18},
  number    = {1},
  year      = {2022},
  url       = {https://doi.org/10.46298/lmcs-18(1:13)2022},
  doi       = {10.46298/lmcs-18(1:13)2022},
  bibsource = {dblp computer science bibliography, https://dblp.org}
}

@inproceedings{ChenSanders09,
  author    = {Yifeng Chen and Jeff W. Sanders},
  editor    = {Ana Cavalcanti and
               Dennis Dams},
  title     = {Unifying Probability with Nondeterminism},
  booktitle = {{FM} 2009: Formal Methods, Second World Congress, Eindhoven, The Netherlands,
               November 2-6, 2009. Proceedings},
  series    = {Lecture Notes in Computer Science},
  volume    = {5850},
  pages     = {467--482},
  publisher = {Springer},
  year      = {2009},
  url       = {https://doi.org/10.1007/978-3-642-05089-3\_30},
  doi       = {10.1007/978-3-642-05089-3\_30},
  bibsource = {dblp computer science bibliography, https://dblp.org}
}

@inproceedings{HartogdeVink99,
  author    = {Jerry den Hartog and Erik P. de Vink},
  editor    = {Christel Baier and
               Michael Huth and
               Marta Z. Kwiatkowska and
               Mark Ryan},
  title     = {Mixing Up Nondeterminism and Probability: a preliminary report},
  booktitle = {First International Workshop on Probabilistic Methods in Verification,
               {PROBMIV} 1998, Indianapolis, Indiana, USA, June 19-20, 1998},
  series    = {Electronic Notes in Theoretical Computer Science},
  volume    = {22},
  pages     = {88--110},
  publisher = {Elsevier},
  year      = {1998},
  url       = {https://doi.org/10.1016/S1571-0661(05)82521-6},
  doi       = {10.1016/S1571-0661(05)82521-6},
  bibsource = {dblp computer science bibliography, https://dblp.org}
}

@inproceedings{MisloveOuaknineWorrell04,
  author    = {Michael W. Mislove and Jo{\"{e}}l Ouaknine and James Worrell},
  editor    = {Flavio Corradini and
               Uwe Nestmann},
  title     = {Axioms for Probability and Nondeterminism},
  booktitle = {Proceedings of the 10th International Workshop on Expressiveness in
               Concurrency, {EXPRESS} 2003, Marseille, France, September 2, 2003},
  series    = {Electronic Notes in Theoretical Computer Science},
  volume    = {96},
  pages     = {7--28},
  publisher = {Elsevier},
  year      = {2003},
  url       = {https://doi.org/10.1016/j.entcs.2004.04.019},
  doi       = {10.1016/j.entcs.2004.04.019},
  bibsource = {dblp computer science bibliography, https://dblp.org}
}

@inproceedings{Mislove00,
  author    = {Michael W. Mislove},
  editor    = {Catuscia Palamidessi},
  title     = {Nondeterminism and Probabilistic Choice: Obeying the Laws},
  booktitle = {{CONCUR} 2000 - Concurrency Theory, 11th International Conference,
               University Park, PA, USA, August 22-25, 2000, Proceedings},
  series    = {Lecture Notes in Computer Science},
  volume    = {1877},
  pages     = {350--364},
  publisher = {Springer},
  year      = {2000},
  url       = {https://doi.org/10.1007/3-540-44618-4\_26},
  doi       = {10.1007/3-540-44618-4\_26},
  bibsource = {dblp computer science bibliography, https://dblp.org}
}

@article{KeimelPlotkin17,
  author    = {Klaus Keimel and Gordon D. Plotkin},
  title     = {Mixed powerdomains for probability and nondeterminism},
  journal   = {Log. Methods Comput. Sci.},
  volume    = {13},
  number    = {1},
  year      = {2017},
  url       = {https://doi.org/10.23638/LMCS-13(1:2)2017},
  doi       = {10.23638/LMCS-13(1:2)2017},
  bibsource = {dblp computer science bibliography, https://dblp.org}
}

@article{VaraccaWinskel06,
  author    = {Daniele Varacca and Glynn Winskel},
  title     = {Distributing probability over non-determinism},
  journal   = {Math. Struct. Comput. Sci.},
  volume    = {16},
  number    = {1},
  pages     = {87--113},
  year      = {2006},
  url       = {https://doi.org/10.1017/S0960129505005074},
  doi       = {10.1017/S0960129505005074},
  bibsource = {dblp computer science bibliography, https://dblp.org}
}

@inproceedings{WangHoffmannReps19,
  author    = {Di Wang and Jan Hoffmann and Thomas W. Reps},
  editor    = {Barbara K{\"{o}}nig},
  title     = {A Denotational Semantics for Low-Level Probabilistic Programs with
               Nondeterminism},
  booktitle = {Proceedings of the Thirty-Fifth Conference on the Mathematical Foundations
               of Programming Semantics, {MFPS} 2019, London, UK, June 4-7, 2019},
  series    = {Electronic Notes in Theoretical Computer Science},
  volume    = {347},
  pages     = {303--324},
  publisher = {Elsevier},
  year      = {2019},
  url       = {https://doi.org/10.1016/j.entcs.2019.09.016},
  doi       = {10.1016/j.entcs.2019.09.016},
  bibsource = {dblp computer science bibliography, https://dblp.org}
}

@article{Affeldt21,
  author    = {Reynald Affeldt and Jacques Garrigue and David Nowak and Takafumi Saikawa},
  title     = {A trustful monad for axiomatic reasoning with probability and nondeterminism},
  journal   = {J. Funct. Program.},
  volume    = {31},
  pages     = {e17},
  year      = {2021},
  url       = {https://doi.org/10.1017/S0956796821000137},
  doi       = {10.1017/S0956796821000137},
  bibsource = {dblp computer science bibliography, https://dblp.org}
}

@InProceedings{MciverCohenMorgan06,
author="McIver, A. K.
and Cohen, E.
and Morgan, C. C.",
editor="Schmidt, Renate A.",
title="Using Probabilistic Kleene Algebra for Protocol Verification",
booktitle="Relations and Kleene Algebra in Computer Science",
year="2006",
publisher="Springer Berlin Heidelberg",
address="Berlin, Heidelberg",
pages="296--310",
isbn="978-3-540-37874-7"
}

@inproceedings{QWGW08,
  author       = {Rui Qiao and
                  Yuan Wang and
                  Xinyan Gao and
                  Jinzhao Wu},
  title        = {Operational semantics of probabilistic Kleene algebra with tests},
  booktitle    = {Proceedings of the 13th {IEEE} Symposium on Computers and Communications
                  {(ISCC} 2008), July 6-9, Marrakech, Morocco},
  pages        = {706--713},
  publisher    = {{IEEE} Computer Society},
  year         = {2008},
  url          = {https://doi.org/10.1109/ISCC.2008.4625616},
  doi          = {10.1109/ISCC.2008.4625616},
  bibsource    = {dblp computer science bibliography, https://dblp.org}
}

@InProceedings{Segala95,
author="Segala, Roberto",
editor="Lee, Insup
and Smolka, Scott A.",
title="A compositional trace-based semantics for probabilistic automata",
booktitle="CONCUR '95: Concurrency Theory",
year="1995",
publisher="Springer Berlin Heidelberg",
address="Berlin, Heidelberg",
pages="234--248",
isbn="978-3-540-44738-2"
}

@book{Hermanns02,
  author       = {Holger Hermanns},
  title        = {Interactive Markov Chains: The Quest for Quantified Quality},
  series       = {Lecture Notes in Computer Science},
  volume       = {2428},
  publisher    = {Springer},
  year         = {2002},
  url          = {https://doi.org/10.1007/3-540-45804-2},
  doi          = {10.1007/3-540-45804-2},
  isbn         = {3-540-44261-8},
  bibsource    = {dblp computer science bibliography, https://dblp.org}
}

@article{MGCM08,
author = {A.K. McIver and C. Gonzalia and E. Cohen and C.C. Morgan},
title = {Using probabilistic Kleene algebra pKA for protocol verification},
journal = {The Journal of Logic and Algebraic Programming},
volume = {76},
number = {1},
pages = {90-111},
year = {2008},
note = {Relations and Kleene Algebras in Computer Science},
issn = {1567-8326},
doi = {https://doi.org/10.1016/j.jlap.2007.10.005},
url = {https://www.sciencedirect.com/science/article/pii/S1567832607000860}
}

@InProceedings{MRS11,
author="McIver, Annabelle
and Rabehaja, Tahiry M.
and Struth, Georg",
editor="de Swart, Harrie",
title="On Probabilistic Kleene Algebras, Automata and Simulations",
booktitle="Relational and Algebraic Methods in Computer Science",
year="2011",
publisher="Springer Berlin Heidelberg",
address="Berlin, Heidelberg",
pages="264--279",
isbn="978-3-642-21070-9"
}

@InProceedings{FTN08,
author="Furusawa, Hitoshi
and Tsumagari, Norihiro
and Nishizawa, Koki",
editor="Berghammer, Rudolf
and M{\"o}ller, Bernhard
and Struth, Georg",
title="A Non-probabilistic Relational Model of Probabilistic Kleene Algebras",
booktitle="Relations and Kleene Algebra in Computer Science",
year="2008",
publisher="Springer Berlin Heidelberg",
address="Berlin, Heidelberg",
pages="110--122",
isbn="978-3-540-78913-0"
}

@article{Rabin63,
author = {Michael O. Rabin},
title = {Probabilistic automata},
journal = {Information and Control},
volume = {6},
number = {3},
pages = {230-245},
year = {1963},
issn = {0019-9958},
doi = {https://doi.org/10.1016/S0019-9958(63)90290-0},
url = {https://www.sciencedirect.com/science/article/pii/S0019995863902900}
}

@inproceedings{DashStaton21a,
author="Swaraj Dash and Sam Staton",
editor="David I. Spivak and Jamie Vicary",
title="A Monad for Probabilistic Point Processes",
booktitle="Applied Category Theory 2020 (ACT2020)",
year="2021",
publisher="Open Publishing Association",
series="EPTCS",
volume="333",
pages="19--32",
doi = "10.4204/EPTCS.333.2"
}

@inproceedings{DashStaton21b,
author="Swaraj Dash and Sam Staton",
editor="Ana Sokolova",
title="Monads for Measurable Queries in Probabilistic Databases",
booktitle="MFPS",
year="2021",
publisher="Open Publishing Association",
series="EPTCS",
volume="351",
pages="34--50",
doi = "10.4204/EPTCS.351.3"
}

@phdthesis{Dash23,
author = "Swaraj Dash",
title = "A Monadic Theory of Point Processes",
school = "Oxford University",
year = "2023"
}

@article{Sokolova11,
 author = {Ana Sokolova},
 title = {Probabilistic systems coalgebraically: A survey},
 journal = {Theoretical Computer Science},
 volume = {412},
 pages = {5095--5110},
 year = {2011}
}

@inproceedings{DArgenio98,
 author={P. D’Argenio and H. Hermanns and J.-P. Katoen},
 booktitle={Proc. PROBMIV’98}, 
 title={On generative parallel composition},
 series={Electronic Notes in Theoretical Computer Science},
 year={1998},
 volume={22},
 pages={105--122}
}

@inproceedings{Glabbeek90,
 author={R.v. Glabbeek and S. Smolka and B. Steffen and C. Tofts},
 booktitle={Proc. LICS, IEEE}, 
 title={Reactive, generative, and stratified models of probabilistic processes}, 
 year={1990},
 pages={130--141}
}

@article{Hansson94,
 author = {H. Hansson},
 title = {Time and probability in formal design of distributed systems},
 journal = {Real-Time Safety Critical Systems},
 volume = {1},
 year = {1994}
}

@article{LarsenSkou91,
 author = {K. Larsen and A. Skou},
 title = {Bisimulation through probabilistic testing},
 journal = {Information and Computation},
 volume = {94},
 pages = {1--28},
 year = {1991}
}

@article{PnueliZuck93,
 author = {A. Pnueli and L. Zuck},
 title = {Probabilistic verification},
 journal = {Information and Computation},
 volume = {103},
 pages = {1--29},
 year = {1993}
}

@phdthesis{SegalaPhD95,
author = {R. Segala},
title = {Modeling and verification of randomized distributed real-time systems},
school = {MIT},
year = {1995}
}

@inproceedings{SegalaLynch94,
 author={R. Segala and N. Lynch},
 booktitle={Proc. CONCUR}, 
 title={Probabilistic simulations for probabilistic processes},
 series={LNCS},
 publisher={Springer},
 year={1994},
 volume={836},
 pages={481--496}
}

@inproceedings{Vardi85,
 author={M. Vardi},
 booktitle={Proc. FOCS, IEEE}, 
 title={Automatic verification of probabilistic concurrent finite state programs}, 
 year={1985},
 pages={327--338}
}

@phdthesis{Hatefi17,
title={Finite horizon analysis of Markov automata},
author={Hatefi Ardakani, Hassan},
school={Universit\"at des Saarlandes},
doi={http://dx.doi.org/10.22028/D291-26690},
year={2016}
}

@article{Rennela16,
  author       = {Mathys Rennela},
  title        = {Convexity and Order in Probabilistic Call-by-Name {FPC}},
  journal      = {CoRR},
  volume       = {abs/1607.04332},
  year         = {2016},
  url          = {http://arxiv.org/abs/1607.04332},
  eprinttype    = {arXiv},
  eprint       = {1607.04332},
  bibsource    = {dblp computer science bibliography, https://dblp.org}
}

@techreport{ZKSTa24,
  author="Noam Zilberstein and Dexter Kozen and Alexandra Silva and Joseph Tassarotti",
  title="A Demonic Outcome Logic for Randomized Nondeterminism",
  month="October",
  year="2024",
  institution="Cornell University",
  number="{\url{https://arxiv.org/abs/2410.22540}}",
  note="POPL 2025, to appear"
}

@inproceedings{FKMRS16,
  author       = {Nate Foster and
                  Dexter Kozen and
                  Konstantinos Mamouras and
                  Mark Reitblatt and
                  Alexandra Silva},
  editor       = {Peter Thiemann},
  title        = {Probabilistic NetKAT},
  booktitle    = {25th European Symposium on Programming,
                  {ESOP} 2016, Eindhoven, The Netherlands,
                  April 2-8, 2016},
  series       = {Lecture Notes in Computer Science},
  volume       = {9632},
  pages        = {282--309},
  publisher    = {Springer},
  year         = {2016},
  url          = {https://doi.org/10.1007/978-3-662-49498-1\_12},
  doi          = {10.1007/978-3-662-49498-1\_12},
  bibsource    = {dblp computer science bibliography, https://dblp.org}
}

@article{K97c,
  title="Kleene algebra with tests",
  author="Dexter Kozen",
  journal="ACM Trans. Programming Languages and Systems (TOPLAS'97)",
  volume="19",
  number="3",
  month="May",
  year="1997",
  pages="427--443",
  doi="10.1145/256167.256195"
}

@inproceedings{AFGJKSW13a,
  author="Carolyn Jane Anderson and Nate Foster and Arjun Guha and Jean-Baptiste Jeannin and Dexter Kozen and Cole Schlesinger and David Walker",
  title="{NetKAT}: Semantic Foundations for Networks",
  booktitle="Proc. 41st ACM SIGPLAN-SIGACT Symp. Principles of Programming Languages (POPL'14)",
  organization="ACM",
  pages="113--126",
  address="San Diego, California, USA",
  month="January",
  year="2014"
}

@inproceedings{SFHKKS20a,
  author="Steffen Smolka and Nate Foster and Justin Hsu and Tobias Kapp\'e and Dexter Kozen and Alexandra Silva",
  title="Guarded {K}leene Algebra with Tests: Verification of Uninterpreted Programs in Nearly Linear Time",
  booktitle="Proc. 47th ACM SIGPLAN Symp. Principles of Programming Languages (POPL'20)",
  organization="ACM",
  pages="61:1--28",
  address="New Orleans",
  month="January",
  year="2020"
}

@inproceedings{KS24a,
  author="Dexter Kozen and Alexandra Silva",
  title="Multisets and Distributions",
  month="May",
  year="2024",
  booktitle="Logics and Type Systems in Theory and Practice: Essays Dedicated to Herman Geuvers on The Occasion of His 60th Birthday",
  editor="Venanzio Capretta and Robbert Krebbers and Freek Wiedijk",
  publisher="Springer LNCS",
  volume="14560",
  pages="168--187"
}

@techreport{OMK24a,
  author="Shawn Ong and Stephanie Ma and Dexter Kozen",
  title="Probability and Angelic Nondeterminism with Multiset Semantics",
  month="December",
  year="2024",
  institution="Cornell University",
  number="{\url{https://arxiv.org/abs/2412.06754}}"
}

@inproceedings{OMK25a,
  author="Shawn Ong and Stephanie Ma and Dexter Kozen",
  title="Probabilistic {K}leene Algebra with Angelic Nondeterminism",
  booktitle="Proc. 46th ACM SIGPLAN Conference on Programming Language Design and Implementation (PLDI 2025)",
  organization="ACM",
  address="Seoul",
  month="June",
  year="2025"
}

@inproceedings{JFKKSSW25a,
  author="Jules Jacobs and Nate Foster and Tobias Kappé and Dexter Kozen and Lily Saada and Alexandra Silva and Jana Wagemaker",
  title="{StacKAT}: Infinite State Network Verification",
  booktitle="Proc. 46th ACM SIGPLAN Conference on Programming Language Design and Implementation (PLDI 2025)",
  organization="ACM",
  address="Seoul",
  month="June",
  year="2025"
}

\begin{toappendix}
\section{\PKA\ Automata}
\label{apx:automata}

\begin{figure}[h]
\centering
\begin{tikzpicture}[->, >=stealth', node distance=2cm, auto, bend angle=35, scale=0.75]
\small
 \node at (-.3,.3) {$s$};
 \node (top) at (0,0) {$\oplus$};
 \node (amp1) at (-6,-2) {$\&$};
 \path (top) edge node[swap] {$p$} (amp1);
 \node (amp2) at (-2,-2) {$\&$};
 \path (top) edge node[swap, pos=.8] {$q$} (amp2);
 \node (amp3) at (2,-2) {$\&$};
 \path (top) edge node[pos=.8] {$r$} (amp3);
 \node (amp4) at (6,-2) {$\&$};
 \path (top) edge node {$1-(p+q+r)$} (amp4);
 \node (amp1a) at (-7.4,-3.6) {$\skp$};
 \path (amp1) edge (amp1a);
 \node (amp1b) at (-6.4,-3.6) {$\skp$};
 \path (amp1) edge (amp1b);
 \node (amp1c) at (-5.6,-3.6) {$a$};
 \path (amp1) edge (amp1c);
 \node (amp1d) at (-4.8,-3.6) {$b$};
 \path (amp1) edge (amp1d);
 \node (amp1cx) at (-5.6,-4.6) {$s$};
 \path (amp1c) edge (amp1cx);
 \node (amp1dx) at (-4.8,-4.6) {$t$};
 \path (amp1d) edge (amp1dx);
 \node (amp2a) at (-2.8,-3.6) {$\skp$};
 \path (amp2) edge (amp2a);
 \node (amp2b) at (-2,-3.6) {$a$};
 \path (amp2) edge (amp2b);
 \node (amp2c) at (-1.2,-3.6) {$a$};
 \path (amp2) edge (amp2c);
 \node (amp2bx) at (-2,-4.6) {$u$};
 \path (amp2b) edge (amp2bx);
 \node (amp2cx) at (-1.2,-4.6) {$v$};
 \path (amp2c) edge (amp2cx);
 \node (amp3a) at (.2,-3.6) {$\skp$};
 \path (amp3) edge (amp3a);
 \node (amp3b) at (1.2,-3.6) {$\skp$};
 \path (amp3) edge (amp3b);
 \node (amp3c) at (2,-3.6) {$a$};
 \path (amp3) edge (amp3c);
 \node (amp3d) at (2.8,-3.6) {$b$};
 \path (amp3) edge (amp3d);
 \node (amp3e) at (3.6,-3.6) {$b$};
 \path (amp3) edge (amp3e);
 \node (amp3cx) at (2,-4.6) {$s$};
 \path (amp3c) edge (amp3cx);
 \node (amp3dx) at (2.8,-4.6) {$t$};
 \path (amp3d) edge (amp3dx);
 \node (amp3ex) at (3.6,-4.6) {$t$};
 \path (amp3e) edge (amp3ex);
 \node (amp4a) at (4.8,-3.6) {$a$};
 \path (amp4) edge (amp4a);
 \node (amp4b) at (5.6,-3.6) {$a$};
 \path (amp4) edge (amp4b);
 \node (amp4c) at (6.4,-3.6) {$a$};
 \path (amp4) edge (amp4c);
 \node (amp4d) at (7.2,-3.6) {$b$};
 \path (amp4) edge (amp4d);
 \node (amp4ax) at (4.8,-4.6) {$s$};
 \path (amp4a) edge (amp4ax);
 \node (amp4bx) at (5.6,-4.6) {$s$};
 \path (amp4b) edge (amp4bx);
 \node (amp4cx) at (6.4,-4.6) {$t$};
 \path (amp4c) edge (amp4cx);
 \node (amp4dx) at (7.2,-4.6) {$t$};
 \path (amp4d) edge (amp4dx);
\end{tikzpicture}
\caption{Fragment of an automaton (reproduced from \cite{OMK25a})}
\label{fig:coalg}
\end{figure}
In the automaton fragment of Fig.~\ref{fig:coalg}, $s,t,u,v$ are states, $a,b$ are actions, and $p,q,r$ are probabilities. In this example, starting from state $s$,
\begin{itemize}
\item
with probability $p$, $\eps$ is accepted with multiplicity 2, and the automaton transitions to $\mset s$ on input $a$ and to $\mset t$ on input $b$;
\item
with probability $q$, $\eps$ is accepted with multiplicity 1, and the automaton transitions to $\mset{u,v}$ on input $a$ and to $\emptyset$ on input $b$;
\item
with probability $r$, $\eps$ is accepted with multiplicity 2, and the automaton transitions to $\mset s$ on input $a$ and to $\mset{t,t}$ on input $b$; and
\item
with probability $1-(p+q+r)$, $\eps$ is accepted with multiplicity 0 (i.e., not accepted), and the automaton transitions to $\mset{s,s,t}$ on input $a$ and to $\mset t$ on input $b$.
\end{itemize}
From state $s$, the probability that $\eps$ is accepted with multiplicity $0$, $1$, or $2$ is $1-(p+q+r)$, $q$, and $p+r$, respectively. On input $a$, the automaton transitions to $\mset s$, $\mset{u,v}$, or $\mset{s,s,t}$ with probabilities $p+r$, $q$, and $1-(p+q+r)$, respectively.

\clearpage
\end{toappendix}

\begin{toappendix}
\section{\PKA\ Equations}
\label{apx:equations}

\begin{table}[h]
\begin{tabular}{l@{\hspace{1cm}}l}
   $e_1 \amp e_2 = e_2 \amp e_1 $ & commutativity of $\amp$\\
   $e_1 \amp (e_2 \amp e_3) = (e_1 \amp e_2) \amp e_3$ & associativity of $\amp$\\
   $e \amp \fail = e$ & identity of $\amp$ \\
   $e_1 \oplus_r e_2 = e_2 \oplus_{1-r} e_1 $ & skew commutativity of $\oplus$\\
   $(e_1 \oplus_r e_2) \oplus_s e_3 = e_1 \oplus_{rs} (e_2 \oplus_{(s-rs)/(1-rs)} e_3)$ & skew associativity of $\oplus$\\
   $e_1 \oplus_1 e_2 = e_1$ & $\oplus$-elimination\\
   $e \oplus_r e = e$ & idempotence of $\oplus$\\
   $e_1 \cmp (e_2 \cmp e_3) = (e_1 \cmp e_2) \cmp e_3$ & associativity of $\cmp$\\
   $\skp \cmp e = e$ & left identity of $\cmp$ \\
   $e \cmp \skp = e$ & right identity of $\cmp$\\
   $\fail \cmp e = \fail$ & left absorption of $\cmp$\\
   $e \cmp \fail = \fail$ & right absorption of $\cmp$\\
   $(e_1 \amp e_2) \cmp e_3 = (e_1 \cmp e_3) \amp (e_2 \cmp e_3)$ & right distributivity of $\cmp$ over $\amp$\\
   $(e_1\oplus_r e_2)\cmp e_3 = (e_1\cmp e_3)\oplus_r (e_2\cmp e_3)$ & right distributivity of $\cmp$ over $\oplus_r$\\
   $(e_1\oplus_r e_2)\amp e_3 = (e_1\amp e_3)\oplus_r(e_2\amp e_3)$ & right distributivity of $\amp$ over $\oplus_r$\\
   $a\cmp(e_1\amp e_2) = (a\cmp e_1)\amp(a\cmp e_2)$ & atomic left distributivity of $\cmp$ over $\amp$\\
   $a\cmp(e_1\oplus_r e_2) = (a\cmp e_1)\oplus_r(a\cmp e_2)$ & atomic left distributivity of $\cmp$ over $\oplus_r$\\
   $d = e[d/x]\ \Iff\ d = \fix xe$ & fixpoint\\
\end{tabular}
\caption{\PKA\ equations (reproduced from \cite{OMK25a})}
\label{tab:equations}
\end{table}
\end{toappendix}

\appendix

\end{document}